\documentclass[a4paper,onecolumn,11pt,accepted=2026-08-12]{quantumarticle}
\pdfoutput=1
\usepackage[margin=1in]{geometry}

\usepackage[sort&compress, numbers]{natbib}
\usepackage{physics}
\usepackage{authblk}  
\usepackage{float}

\usepackage{graphicx}
\usepackage[dvipsnames]{xcolor}
\usepackage{rotating}
\usepackage{amsmath,amssymb,graphics,amsthm,isomath}
\usepackage{amsfonts,dsfont,mathtools}
\usepackage{bbm}
\usepackage{bm}
\usepackage{array}
\usepackage{appendix}
\newcolumntype{P}[1]{>{\centering\arraybackslash}p{#1}}
\usepackage{multirow}
\usepackage{physics}
\usepackage{braket}
\usepackage{adjustbox}
\usepackage{stmaryrd}
\usepackage{blkarray, bigstrut}
\usepackage{tikz-cd}
\usepackage{quantikz}
\usepackage{colortbl}
\usepackage[colorlinks=true, urlcolor=violet, linkcolor=blue, citecolor=blue, hyperindex=true, linktocpage=true]{hyperref}
\usepackage[capitalise,compress]{cleveref}

\allowdisplaybreaks

\newtheorem{thm}{Theorem}
\numberwithin{thm}{section}
\newtheorem{cor}[thm]{Corollary}
\newtheorem{lem}[thm]{Lemma}

\newcommand{\ycmark}{{\color{orange}\ding{51}}}
\renewcommand{\thetable}{\arabic{table}}
\renewcommand{\thesection}{\arabic{section}}
\renewcommand{\thesubsection}{\thesection.\arabic{subsection}}
\renewcommand{\thesubsubsection}{\thesubsection.\arabic{subsubsection}}
\makeatletter
\renewcommand{\p@subsection}{}
\renewcommand{\p@subsubsection}{}
\makeatother

\usepackage{xcolor}
\usepackage{mathtools}

\newcommand\bea{\begin{eqnarray}}
\newcommand\eea{\end{eqnarray}}
\newcommand\be{\begin{equation}}
\newcommand\ee{\end{equation}}
\newcommand\bes{\begin{subequations}}
\newcommand\ees{\end{subequations}}
\newcommand\bed{\begin{displaymath}}
\newcommand\eed{\end{displaymath}}
\newcommand\beal{\begin{aligned}}
\newcommand\eeal{\end{aligned}}
\newcommand\bew{\begin{widetext}}
\newcommand\eew{\end{widetext}}
\newcommand\beit{\begin{itemize}}
\newcommand\eeit{\end{itemize}}
\def\bea{\begin{array}}
\def\eea{\end{array}}
\newcommand\been{\begin{enumerate}}
\newcommand\eeen{\end{enumerate}}

\usepackage{verbatim}

\usepackage{tikz}

\definecolor{red1}{rgb}{0.76, 0.23, 0.13}
\definecolor{green1}{rgb}{0.0, 0.5, 0.0}
\usepackage{pifont} 
\newcommand{\cmark}{\textcolor{green1}{\ding{51}}}
\newcommand{\xmark}{\textcolor{red1}{\ding{55}}}

\usepackage{booktabs}

\begin{document}

\title{Simple logical quantum computation with concatenated symplectic double codes}

\author[1,2]{Noah Berthusen}
\email{noah.berthusen@quantinuum.com}
\author[2]{Elijah Durso-Sabina}
\affil[1]{Joint Center for Quantum Information and Computer Science, NIST/University of Maryland, College Park, Maryland 20742, USA}
\affil[2]{Quantinuum, Broomfield, CO 80021, USA}

\maketitle
\renewcommand{\thefootnote}{\fnsymbol{footnote}}
\renewcommand{\thefootnote}{\arabic{footnote}}

\begin{abstract}
There have been significant recent advances in constructing theoretical and practical quantum error correcting codes that function well as quantum memories; however, performing fault-tolerant logical gates on these codes is less studied, and the protocols that do exist often require significant complexity. Building off the symplectic double construction, we investigate \textit{concatenated symplectic double} codes, which have a rich set of logical gates implementable using only physical single-qubit gates and qubit relabeling. Combined with injected fold-transversal gates, the full Clifford group on a single codeblock is achieved through a functionally simple circuit. We perform circuit-level simulations of state preparation and quantum error correction on these codes and show that they have promising performance at near state-of-the-art physical error rates. As such, we argue that concatenated symplectic double codes are strong contenders as the underlying computational code on medium- to large-scale quantum computers.
\end{abstract}

\tableofcontents

\section{Introduction}

To achieve the computational speed-ups offered by quantum computers, it is widely believed that quantum error correction will be required to manage errors throughout the execution of the circuit. Considering only stabilizer codes~\cite{gottesman1997}, there are roughly three main categories of quantum error correcting codes being considered for large-scale fault-tolerant quantum computation: topological codes like the surface code~\cite{bravyi1998quantum, Kitaev_2003} have been widely used, particularly due to their amenability to planar architectures; however, they suffer from poor encoding rates which imply significant overheads at scale.
Quantum low-density parity-check (qLDPC) codes~\cite{breuckmann2021} offer better encoding rates and distances than topological codes but require qubit connectivity that some quantum computing hardware cannot support. Like topological codes, concatenated codes have a rich history, with the first proofs of the threshold theorem utilizing them~\cite{Knill_1998, Aharonov_Ben-Or_1999, sommers2025}. Recently, concatenated codes with high encoding rates have been developed~\cite{Goto_2024, Yamasaki_2024, Yoshida_2025}. These codes have the same connectivity issues as qLDPC codes as well as much higher check weights. Nonetheless, there is evidence to suggest they could be competitive in practice.

Besides ease of implementation, topological codes have the benefit of simple logical operations through transversal gates and lattice surgery, leading to fairly straightforward and reasonably efficient logical circuits. 
For the latter two code categories, much less is known about efficient ways to perform encoded gates fault-tolerantly. Perhaps the state of the art for qLDPC codes is the recently introduced SHYPS codes~\cite{malcolm2025}, which can execute any $m$-Clifford operation in at most $O(m)$ syndrome extraction rounds.
For other qLDPC codes, e.g. bivariate bicycle codes~\cite{bravyi2023highthreshold}, computation is facilitated by Pauli product measurements and implemented through generalized lattice surgery~\cite{Cowtan_2024, cross2024}. Typically $O(m)$ syndrome extraction rounds and $O(d)$ space and time overhead is required; although, recent improvements in decoding and constructing more efficient ancilla systems have shown that the required spacetime overhead can be reduced to near-constant~\cite{cain2024, zheng2025}. The story is similar for concatenated codes, where distillation and state injection are required for many logical gadgets, resulting in significant (practical) space and time overheads. We make progress on this front by presenting finite-sized concatenated codes with an easily implementable logical Clifford group.

Recently, Ref.~\cite{berthusen2025} investigated $C_4$-concatenated hypergraph product codes. With the additional structure of the concatenation scheme---coming from a particular logical-physical qubit assignment---they found that some fold-transversal~\cite{Breuckmann_2024} operations on the original hypergraph product code could be upgraded to SWAP-transversal on the $C_4$-concatenated code. Motivated by this result, we investigate a code construction which we call \textit{concatenated symplectic double (CSD)} codes. A CSD code is obtained by concatenating any $[[n,2,d]]$ inner code $Q$ along a ZX-duality~\cite{Breuckmann_2024} of an outer code constructed as the symplectic double of a non-CSS code~\cite{burton2024}. Depending on the available logical gates of the inner code $Q$, logical gates on the full concatenated code may be easier to implement. In this work we study the $C_4$~\cite{Vaidman_1996, Grassl_1997} inner code in detail: by construction, all fold-transversal gates on the symplectic double code are upgraded to be SWAP-transversal on the concatenated code, allowing for a significant number of logical gates to be implemented using only physical single-qubit gates and qubit relabeling. Adding fold-transversal logical gates through state injection allows us to generate the full Clifford group on every logical qubit in a single codeblock. Notably, every logical Clifford circuit on a single computational codeblock has the same, simple, physical circuit, and can be combined with error correction in such a way that functionally, the only operation being performed is QEC facilitated by teleportation. Such a gateset is well suited for quantum computers such as neutral atoms~\cite{Bluvstein_2023} and ion traps~\cite{Moses_2023} that can use qubit movement to facilitate relabeling.

The CSD construction is very general, and any non-CSS code can be used as the seed code and any $[[n,2,d]]$ code as the inner code. While a family of non-CSS codes or a family of inner codes would each give rise to a family of CSD codes, we choose specific non-CSS seed codes and study the $C_4$ inner code in detail, yielding $C_4$-CSD codes with good parameters, an expressive gateset, and competitive circuit-level performance. The paper is structured as follows: in Section~\ref{sec:code_construction} we present the general CSD construction and then specialize to the $C_4$ inner code, presenting the resulting $C_4$-CSD codes and listing several promising instances. In Section~\ref{sec:gadgets} we describe how the $C_4$-CSD construction yields codes with many SWAP-transversal logical gates. We then provide an example with the $[[16,4,4]]$ $C_4$-CSD code illustrating how these gates arise, and we show how an example logical circuit could be compiled into the available gateset. We additionally present state preparation and quantum error correction gadgets tailored for $C_4$-CSD codes. In Section~\ref{sec:zx_prep} and Section~\ref{sec:performance} we perform circuit-level simulations of state preparation and error-correction, respectively, and show that $C_4$-CSD codes display promising performance as quantum memories. We argue from the numerical simulation results and the unique compilation strategy that these codes could enable high-fidelity logical computation. We conclude in Section~\ref{sec:discussion_outlook} by posing several open questions.


\section{Concatenated symplectic double codes}
\label{sec:code_construction}

The symplectic double~\cite{burton2024} was recently introduced as a way to construct CSS codes with fault-tolerant Clifford gates. 
When the base code is a CSS code $C$, the symplectic double construction yields a trivial double cover of $C$; that is, we get a $[[2n,2k,d]]$ code obtained from stacking two copies of $C$, $C \oplus C$. 
The more interesting application of the symplectic double is when the base code $C$ is not CSS. In this case, the double cover is not trivially two copies of $C$, but is instead a new CSS code which can be thought of as a trivial double cover with a global twist. The resulting code is called the symplectic double of $C$ and is denoted $\mathfrak{D}(C)$. In the following, we focus on seed codes $C$ which are not CSS. 

The parity check matrix of a $[[n,k,d]]$ non-CSS code $C$ is represented as a $m \times 2n$ parity check matrix $H = (H_X ~|~ H_Z)$. Then, the symplectic double of $C$, $\mathfrak{D}(C)$, is defined to be the CSS code with the following symplectic parity check matrix:

\begin{equation}
\mathfrak{D}(H) = \begin{pmatrix}
\begin{array}{cc|cc}
H_X & H_Z & 0 & 0 \\
0 & 0 & H_Z & H_X \\
\end{array}
\end{pmatrix}.
\label{eq:pcm_double}
\end{equation}

\begin{thm}[Theorem 3.2 of~\cite{burton2024}]
    \label{thm:double}
    Given a quantum code $C$ with parameters $[[n,k,d]]$, $\mathfrak{D}(C)$ is a CSS quantum code with parameters $[[2n,2k,\ge d]]$.
\end{thm}

The resulting symplectic double code $\mathfrak{D}(C)$ acquires a symmetry from the double cover called a \textit{ZX}-duality. A $ZX$-duality of a CSS code is a permutation $\tau: n \rightarrow n$ that swaps the $X$- and $Z$- sectors of the code; that is, it takes $X$-type stabilizers and logicals to $Z$-type stabilizers and logicals, and vice versa.  
For symplectic double codes, the double cover provides the fixed-point-free involutory ZX-duality $\tau$ with orbits of the form $\{i, i+n\}_{i=1, ..., n}.$ This can be interpreted as a symmetry between the copies of $C$ in $\mathfrak{D}(C)$. 

\begin{figure}
    \centering
    \includegraphics[width=\linewidth]{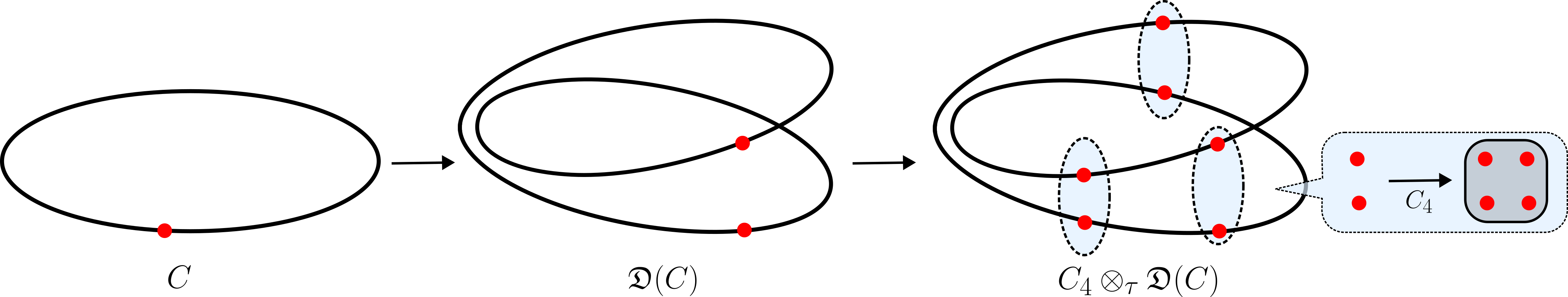}
    \caption{\small Schematic diagram of constructing a $C_4$-CSD code. A symplectic double code $\mathfrak{D}(C)$ is obtained by taking the symplectic double cover of a non-CSS code $C$, Eq.~\eqref{eq:pcm_double}. The concatenated symplectic double code $C_4 \otimes \mathfrak{D}(C)$ is then obtained by taking pairs of physical qubits of $\mathfrak{D}(C)$ related through the ZX-duality $\tau$ and encoding them as the logical qubits in the same $C_4$ codeblock. Figure adapted from Fig. 1 of Ref.~\cite{burton2024}.}
    \label{fig:schematic}
\end{figure}

It is using this ZX-duality $\tau$ that we now construct a \textit{concatenated symplectic double} code. In general, a CSD code is a two-level concatenated code with $\mathfrak{D}(C)$ as the outer code and any $[[n',2,d']]$ stabilizer code $Q$ as the inner code. The concatenation is performed following the method described in Section 3.5 of Ref.~\cite{gottesman1997} and Procedure 1 of Ref.~\cite{berthusen2025}, in which the logical qubits of several inner code blocks act as the physical qubits of the outer stabilizer code. The ZX-duality $\tau$ informs the mapping of inner logical qubits to outer physical qubits: since $Q$ encodes exactly two logical qubits, and $\tau$ pairs each physical qubit $i$ of $\mathfrak{D}(C)$ with qubit $i+n$, we assign physical qubits $\{i, i+n\}_{i=1, ..., n}$ to the same $Q$ codeblock. That is, the two $\tau$-related physical qubits in each orbit become the two logical qubits of a single inner codeblock. Since $\tau$ is fixed-point-free, every physical qubit of $\mathfrak{D}(C)$ is assigned to a $Q$ block in this way. We call the resulting code a concatenated symplectic double code and denote it $Q \otimes_\tau \mathfrak{D}(C)$.

The motivation for this construction is that logical gates of the outer code $\mathfrak{D}(C)$ may be lifted, via the concatenation structure, to logical gates on the full code $Q \otimes_\tau \mathfrak{D}(C)$. Specifically, as described in Section~\ref{sec:gadgets}, logical gates of $\mathfrak{D}(C)$ that act within $\tau$-orbits can be implemented by applying the corresponding logical operation at the level of the inner $Q$ codeblocks. Depending on the properties of $Q$, this may upgrade the fault-tolerance of those logical gates. The following lemma bounds the parameters of the general CSD code.

\begin{lem}
\label{lem:csd}
  Given a quantum code $C$ with parameters $[[n,k,d]]$ and an $[[n',2,d']]$ inner code $Q$, $Q \otimes_\tau \mathfrak{D}(C)$ is a CSS quantum code with parameters $[[nn', 2k, dd' \ge D \ge dd'/2]]$.
\end{lem}
\begin{proof}
    From Theorem~\ref{thm:double}, the symplectic double code has parameters $[[2n,2k\ge2d]]$. The concatenated code $Q \otimes_\tau \mathfrak{D}(C)$ has $n$ inner codeblocks of $n'$ physical qubits each, giving $nn'$ total physical qubits and $2k$ logical qubits. The upper bound on the distance can be obtained if a minimum-weight logical operator of $\mathfrak{D}(C)$ has support on $d$ blocks of $Q$, in which case the corresponding logical of $Q \otimes_\tau \mathfrak{D}(C)$ has weight $dd'$. The distance can be lowered if both logical operators from $Q$ are part of the same logical operator from $\mathfrak{D}(C)$; since the distance of the underlying code is $d'$, the weight of the product of the two logicals must be at least $d'$ as well. This can happen on $d/2$ blocks of $Q$ in the worst case, giving the distance lower bound of $dd'/2$. 
\end{proof}

\subsection{$C_4$-CSD codes}

In this work, we study the special case where the inner code $Q$ is the $C_4$ code~\cite{Vaidman_1996, Grassl_1997}, yielding the self-dual $C_4$-CSD code $C_4 \otimes_\tau \mathfrak{D}(C)$. We define the $C_4$ code with two stabilizer generators $XXXX$, $ZZZZ$ and two pairs of anticommuting logical operators $\overline{X}_1 = XXII$, $\overline{Z}_1 = ZIZI$, $\overline{X}_2 = XIXI$, $\overline{Z}_2 = ZZII$. Note that there are other valid choices for the $C_4$ logical operators which give rise to slightly different logical gate implementations, see Ref.~\cite{sayginel2024fault}. The entire procedure to construct a $C_4$-CSD code is illustrated in Fig.~\ref{fig:schematic}. We note that this construction was mentioned in Ref.~\cite{burton2024} but not extensively studied. $C_4$-concatenated codes with alternative outer codes have also been studied previously~\cite{criger2016, berthusen2025}. 

The fact that $C_4 \otimes_\tau \mathfrak{D}(C)$ is self-dual can be seen as follows: consider Eq.~\eqref{eq:pcm_double}, and let us construct the parity check matrix for $C_4 \otimes \mathfrak{D}(C)$. Each row in $\mathfrak{D}(C)$ will become a stabilizer generator in the concatenated code. More explicitly, a 1 in the $i$th column will be replaced with $\overline{X}_1$ in the $i$th $C_4$ block if $i < n$ and $\overline{X}_2$ if $n <= i < 2n$. Similarly, a 1 in the $(i+2n)$th column will be replaced with $\overline{Z}_1$ in the $i$th $C_4$ block if $i < n$ else $\overline{Z}_2$. This mapping is illustrated below in Eq.~\eqref{eq:self-dual}. 
\begin{equation}
\includegraphics[width=0.37\linewidth]{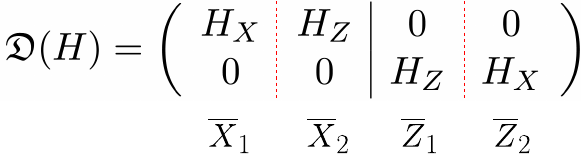}
\label{eq:self-dual}
\end{equation}
From this, we can see that $H_X$ tells us to place $\overline{X}_1$ and $\overline{Z}_2$ in the same $C_4$ blocks, and similarly with $H_Z$, $\overline{X}_2$ and $\overline{Z}_1$. Since we have chosen $C_4$ logical operators such that $\overline{X}_1 = \overline{Z}_2$ and $\overline{X}_2 = \overline{Z}_1$ (up to a Hadamard transform), we can see that the resulting concatenated generators will have exactly the same support as well as even weight. Combined with the fact that the $C_4$ code itself is self-dual, we conclude that $C_4 \otimes_\tau \mathfrak{D}(C)$ is self-dual.


While viewing these codes as concatenated codes yields insights into the available logical operators, they are, in fact, equivalent to a class of self-dual CSS codes obtained through the use of Majorana fermion codes~\cite{Bravyi_2010_ferm}. From this equivalence we can tighten the distance bounds of Lemma~\ref{lem:csd} when applied to $C_4$-CSD codes. Here we explain the correspondence.

\begin{lem}[Corollary 1 of~\cite{Bravyi_2010_ferm}]
  Any $[[n,k,d]]$ stabilizer code $\mathcal{S}$ can be mapped onto a $[[4n,2k,2d]]$ self-dual CSS code.  
\end{lem}
The construction is as follows: we first map $\mathcal{S}$ onto a $[[4n,k,2d]]$ Majorana fermion code $\mathcal{S}_{\text{maj}}$. For each qubit $i$ in $\mathcal{S}$, we associate four Majorana fermion modes $b_i^xb_i^yb_i^zc_i$. Then for each stabilizer generator in $\mathcal{S}$, we replace each physical Pauli operator with creation/annihilation operators as 
\begin{equation}
    X_i \rightarrow ib_i^xc_i, \quad Y_i \rightarrow ib_i^yc_i, \quad Z_i \rightarrow ib_i^zc_i.
    \label{eq:majoranas}
\end{equation}
For example, the Majorana fermion code arising from the non-CSS [[4,2,2]] code, Eq.~\eqref{eq:noncss422}, $\mathcal{S} = \langle XZZX, ZXXZ\rangle$, is generated by
\begin{equation}
    \mathcal{S}_{\text{maj}} = \langle b_0^xb_0^yb_0^zc_0, b_1^xb_1^yb_1^zc_1, b_2^xb_2^yb_2^zc_2, b_3^xb_3^yb_3^zc_3, b_0^xc_0b_1^zc_1b_2^zc_2b_3^xc_3, b_0^zc_0b_1^xc_1b_2^xc_2b_3^zc_3 \rangle.
\end{equation}
Taking the resulting Majorana fermion code, we apply Lemma 2 of Ref.~\cite{Bravyi_2010_ferm} to yield a $[[4n,2k,2d]]$ self-dual CSS code $\mathcal{S}'$. Each of the $4n$ creation/annihilation operators is associated with a physical qubit. Then for each stabilizer of $\mathcal{S}_{\text{maj}}$, we associate an $X$- and $Z$-type stabilizer of $\mathcal{S'}$, e.g. $b_0^xb_0^yb_0^zc_0 \rightarrow X_0X_1X_2X_3, ~Z_0Z_1Z_2Z_3$. From this mapping we can see that we obtain generators equal to the two weight-4 $C_4$ generators for each of the $n$ $C_4$ codeblocks. There are an additional $n-k$ generators which correspond to the concatenated generators of $\mathfrak{D}(C)$. When we instead define the $C_4$ logical operators as $\overline{X}_1 = XIIX$, $\overline{Z}_1 = IIZZ$, $\overline{X}_2 = IIXX$, $\overline{Z}_2 = ZIIZ$, the physical-to-logical mappings of Eq.~\eqref{eq:self-dual} match the physical-to-Majorana fermion mapping of Eq.~\eqref{eq:majoranas}. The following corollary arises from the fact that these two constructions are identical.

\begin{cor}
    Given a quantum code $C$ with parameters $[[n,k,d]]$, $C_4 \otimes_\tau \mathfrak{{D}}(C)$ is a self-dual CSS quantum code with parameters $[[4n,2k,2d]]$.
    \label{thm:csd_params}
\end{cor}

\renewcommand{\arraystretch}{1.25}
\setlength{\tabcolsep}{8pt}

\begin{table}[t]
\centering
\begin{tabular}{|c|c|>{\columncolor[RGB]{230, 242, 255}}c|c|c|c|c|}
\hline
$C$  & $\mathfrak{D}(C)$ & $C_4 \otimes_\tau \mathfrak{D}(C)$& $q_{\max}$ & $|G_{\tau}|$  & $|\langle\overline{U}_P(\pi) \rangle |$  & $G \cong \mathcal{C}_k?$                \\
\hline\hline
$[[4,2,2]]$ & $[[8,4,2]]$ & $[[16,4,4]]$ & 8 & 216 & $2^9$ & \cmark \\
\hline
$[[6,4,2]]$  & $[[12,8,2]]$   & $[[24,8,4]]$ & 12 & 25920 & $2^{30}$ & \cmark \\
\hline
$[[8,6,2]]$ & $[[16,12,2]]$ & $[[32,12,4]]$ & 16 & 1451520 & $\ge 2^{63}$ & ?\\
\hline
\hline
$[[5,1,3]]$~\cite{laflamme1996} & $[[10,2,3]]$ & $[[20,2,6]]$ & 8 & 36 & \cellcolor{green!25}$2^3$ & \cmark \\
\hline
$[[8,3,3]]$~\cite{Gottesman_1996} & $[[16,6,3]]$ & $[[32,6,6]]$ & 16 & $1008$ & $2^4$ & \xmark \\
\hline
$[[9,3,3]]$ & $[[18,6,3]]$ & $[[36,6,6]]$ & 12 & 648 & \cellcolor{green!25}$2^{21}$ & \cmark \\
\hline
$[[10,4,3]]$ & $[[20,8,3]]$ & $[[40,8,6]]$ & 12 & 2160 & $2^{29}$ & \cmark \\
\hline
$[[11,5,3]]$ & $[[22,10,3]]$ & $[[44,10,6]]$ & 16 & 2160 & $2^{35}$ & ? \\
\hline
\hline
$[[12,2,4]]$ & $[[24,4,4]]$ & $[[48,4,8]]$ & 12 & 72 & $2^{6}$ & \ycmark \\
\hline
$[[12,4,4]]$ & $[[24,8,4]]$ &  $[[48,8,8]]$ & 16 & 2160 & $\ge 2^{14} $ & \xmark \\
\hline
$[[14,6,4]]$ & $[[28,12,4]]$ & $[[56,12,8]]$ & 16 & 48384 & $\ge 2^{45}$ & ? \\
\hline
\hline
$[[15,3,5]]$~\cite{kanomata2025} & $[[30,6,5]]$ & $[[60,6,10]]$ & 16 & 216 & $\ge 2^4$ & \xmark \\
\hline
$[[18,6,5]]$ & $[[36,12,5]]$ & $[[72,12,10]]$ & 20 & 648 & $\ge 2^7$ & ? \\
\hline
\end{tabular}
\caption{\small Examples of $C_4$-concatenated symplectic double codes. $q_{\max}$ is the maximum stabilizer weight, $|G_\tau|$ is the number of gates implementable SWAP-transversally, and $|\langle\overline{U}_P(\pi) \rangle|$ is the number of fold-transversal gates. Green highlighted cells indicate that the size of $\langle\overline{U}_P(\pi) \rangle$ is maximal, and we indicate lower bounds where necessary. The final column indicates whether $G = \langle G_\tau, \overline{U}_P(\pi) \rangle$ generates the full Clifford group: green check indicates yes, yellow check indicates yes with partial plus state injections, and red $X$ indicates no unless adapters with logical $S$ injections are included.}
\label{tab:codes}
\end{table}

Table~\ref{tab:codes} presents examples of concatenated symplectic double codes, $C_4 \otimes_\tau \mathfrak{D}(C)$ along with the non-concatenated version and the seed non-CSS code. The parity check matrices for the seed codes used in this work are given explicitly in Appendix~\ref{apx:seed_codes}. 
From these and the $ZX$-duality $\tau$, the parity check matrices of the symplectic double and resulting concatenation can be obtained. In Table~\ref{tab:codes} we additionally report the maximum check weight, $q_{\max}$, of $C_4 \otimes_\tau \mathfrak{D}(C)$ and the number of SWAP-transversal gates, the number of fold-transversal gates, and whether these two gatesets generates the full Clifford group.

We make the interesting note that certain instances of the many-hypercubes code (MHC)~\cite{Goto_2024} can be obtained through the $C_4$-CSD construction method. The $[[16,4,4]]$ $C_4$-CSD code obtained by using the seed non-CSS $[[4,2,2]]$ code, Eq.~\eqref{eq:noncss422}, is isomorphic to the $L=2$ $C_4$ MHC code. Higher $L$'s can also be obtained: we first perform a Hadamard transform~\cite{Bonilla_Ataides_2021} on the $[[16,4,4]]$ $C_4$-CSD code to yield a non-CSS code with the same parameters. 
Although unstudied in the original paper, it is possible to construct a MHC code where the different concatenation levels are different Iceberg codes. Such MHC codes are also recoverable through the CSD framework, as long as the lowest level code is the $C_4$ ($[[4,2,2]]$) Iceberg code. Interpreting MHC codes through the $C_4$-CSD framework gives an alternative method of finding valid logical operators.

\section{Fault-tolerant gadgets}
\label{sec:gadgets}

\subsection{Logical Clifford circuits}
\label{sec:log_clifford}

For logical gates, and specifically inter-block logical gates, transversal gates are the gold standard as they prevent the propagation of errors. 
For certain quantum computing modalities, such as neutral atoms or ion-traps, it makes sense to relax the notion of transversality to include SWAP gates, as qubit movement on these devices is low-cost and qubit relabeling can be done in software.
Relaxing further, the notion of fold-transversality has been introduced~\cite{moussa2016, Quintavalle_2023, Breuckmann_2024}. Here, some logical gates are defined with respect to a ZX-duality, $\tau$. A unitary $U$ is then considered fold-transversal if it consists of single-qubit gates and two-qubit gates supported on the orbits $\{(i, \tau(i))\},{i=1,...,n}$ of $\tau$. In general, it is desirable to have logical gates with as strict a notion of transversality as possible, i.e. transversal $\ge$ SWAP-transversal $>$ fold-transversal $\approx$ non-transversal, since less work is required to make them fault-tolerant.

The symplectic double $\mathfrak{D}$ lifts to a homomorphism $\mathfrak{D}': \textrm{Sp}(2n, \mathbb{F}_2) \rightarrow \textrm{Sp}(4n,\mathbb{F}_2)$. In other words, logical gates on $C$ get lifted to logical gates on $\mathfrak{D}(C)$ according to $\mathfrak{D}'$. Note that any Clifford operator on $n$ qubits can be implemented with an element of $\mathrm{Sp}(2n,\mathbb{F}_2)$ and a Pauli correction. The physical operations are lifted according to Theorem 3.7 of Ref.~\cite{burton2024}, and the mappings are shown in Fig.~\ref{fig:mapping}. We note that there is also a lifted version of a CNOT gate; however, we do not make use of such lifted gates in this work. 
Since CNOT and SWAP gates can only take $X$-stabilizers (logicals) to linear combinations of $X$-stabilizers (logicals)---and analogously for $Z$-type operators---it can be seen that all logical gates lifted from $C$ enact some CNOT-type logical operation on $\mathfrak{D}(C)$. Furthermore, it can be seen that lifted logical gates are fold-transversal, as the only physical CNOT gates are applied across the $ZX$-duality described by $\tau$. $\mathfrak{D}'$ informs us of the corresponding physical circuit to apply on $\mathfrak{D}(C)$, with which we can then check the explicit logical action on the symplectic space. The gates on $C$ that we are interested in correspond to automorphisms of the code, which we identify using GAP~\cite{GAP4} and open-source software~\cite{autqec_code} developed in Ref.~\cite{sayginel2024fault}. While the strict definition of automorphism or permutation gates only allows for qubit relabeling~\cite{berthusen2025_2}, we loosen this definition to include logical gates consisting of physical single-qubit gates followed by (a potentially trivial) relabeling of the qubits. This is the definition considered in Ref.~\cite{sayginel2024fault} and is equivalent to the notion of SWAP-transversality. Throughout the paper, we use these two terms interchangeably. As such, transversal intra-block gates and gates such as Hadamard-type fold-transversal gates~\cite{Breuckmann_2024} below are included in our definition of an automorphism gate. 

\begin{figure}
    \centering
    \includegraphics[width=0.35\linewidth]{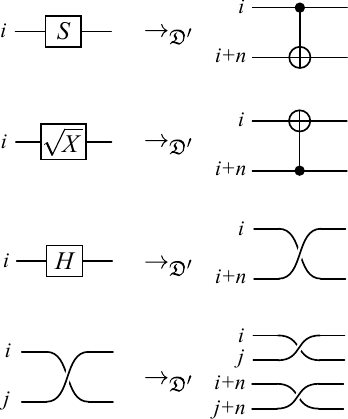}
    \caption{Using $\mathfrak{D}'$ to lift physical operations on $C$ to physical operations on $\mathfrak{D}(C)$. Only CNOT-type logical gates are obtainable through this lifting process.}
    \label{fig:mapping}
\end{figure}

While these are the only logical gates inherited from the automorphism gates of $C$, $\mathfrak{D}(C)$ acquires additional logical operators as a result of the double cover. In particular, the $ZX$-duality $\tau$ enables a Hadamard-type gate:
\begin{equation}
    \overline{H}_{\tau} = \bigotimes_{i=1}^{2n} H_i ~ \bigotimes_{\substack{i=1,...,2n \\ i < \tau(i)}} \text{SWAP}_{i,\tau(i)}.
    \label{eq:hswap}
\end{equation}
This can easily be verified to be a valid logical operator for $\mathfrak{D}(C)$ by considering Eq.~\eqref{eq:pcm_double} and recalling that the $ZX$-duality $\tau$ has the form $\{i, i+n\}_{i=1, ..., n}$. Finally, from Theorem 7 of Ref.~\cite{Breuckmann_2024}, we can say that $\mathfrak{D}(C)$ will have a transversal CZ gate of the form:
\begin{equation}
    \overline{S}_{\tau} = \bigotimes_{\substack{i=1,...,2n \\ i < \tau(i)}} CZ_{i, \tau(i)}.
    \label{eq:s}
\end{equation}
Note that in the most general description of a phase-type fold-transversal gate, additional physical $S$ and $S^\dagger$ gates are applied. These physical gates are not needed for $\mathfrak{D}(C)$ codes as the $ZX$-duality $\tau$ has no invariant qubits.

As symplectic matrices (which can also be thought of as Clifford tableaus describing the action of the logical $X$ and $Z$ operators), the logical actions of these operators take the general form:
\renewcommand\arraystretch{1.25}
\begin{equation}
    \mathfrak{D}'(\overline{L}_i) = 
    \begin{pmatrix}
    \begin{array}{c|c}
        M & 0\\
        \hline
        0 & (M^{-1})^\top  \\
    \end{array}
    \end{pmatrix}, \quad
    \overline{H}_\tau = 
    \begin{pmatrix}
    \begin{array}{c|c}
        0 & M \\
        \hline
        (M^{-1})^\top & 0 \\
    \end{array}
    \end{pmatrix}, \quad 
    \overline{S}_\tau = 
    \begin{pmatrix}
    \begin{array}{c|c}
        I & S \\
        \hline
        0 & I \\
    \end{array}
    \end{pmatrix}, 
    \label{eq:gates}
\end{equation}
where $\overline{L}_i$ is a logical lifted from $C$, $M$ is a $k \times k$ invertible binary matrix, and $S$ is an $k \times k$ symmetric binary matrix with a zero diagonal, $S_{ii} = 0$. These gates on $\mathfrak{D}(C)$ form a group $G_\tau = \langle \mathfrak{D}'(\overline{L}_i), \overline{H}_\tau, \overline{S}_\tau \rangle$ of fold-transversal gates. It may be the case that $\mathfrak{D}(C)$ has its own automorphisms that give rise to additional logical gates; however, there is no guarantee they would be SWAP-transversal and respect $\tau$, see Table~\ref{tab:actions} for an example. We display the size of $G_\tau$ for the codes in Table~\ref{tab:codes}, as calculated with GAP~\cite{GAP4}. All of the aforementioned logical operators in $G_\tau$ remain logical operators of $C_4 \otimes_\tau \mathfrak{D}(C)$ and have the same logical action; however, some of them are now \textit{upgraded} in terms of fault-tolerance. In particular, the concatenation structure of $C_4 \otimes_\tau \mathfrak{D}(C)$ allows us to upgrade all fold-transversal logical gates on $\mathfrak{D}(C)$ to be SWAP-transversal on the concatenated code.

\begin{thm}
    All logical gates from $G_\tau$ are SWAP-transversal on $C_4 \otimes_\tau \mathfrak{D}(C)$
    \label{thm:transversal}
\end{thm}

\begin{figure}[t]
    \centering
    \includegraphics[width=0.5\linewidth]{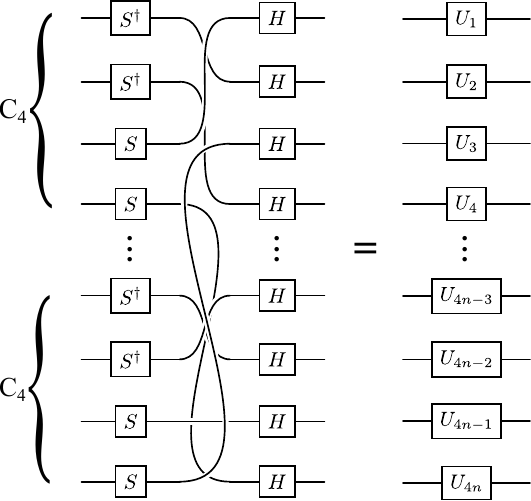}
    \caption{\small An example physical circuit implementing the $\overline{S}_\tau$, permutation, and $\overline{H}_\tau$ (left to right) gates. The entire circuit can be ``untangled'' and compressed to a layer of single-qubit Clifford gates, $U_1  \otimes ...\otimes U_{4n}$, $U_i \in \mathcal{C}_1$.}
    \label{fig:circuit_simplified}
\end{figure}

\begin{proof}
    For our particular concatenation structure, the logical qubits of the $C_4$ codeblocks can be thought of as the `physical' qubits of the $\mathcal{D}(C)$ code. As such, implementing some logical operator in $G_\tau$ on $C_4 \otimes_\tau \mathfrak{D}(C)$ is done by applying the corresponding circuit, (e.g. Fig.~\ref{fig:mapping}, Eq.~\eqref{eq:hswap}, or Eq.~\eqref{eq:s}) to these `physical' qubits. This requires rewriting each physical operation to act on the logical qubits of the $C_4$ codeblocks. The resulting logical action on $C_4 \otimes_\tau \mathfrak{D}(C)$ is the same as on $\mathfrak{D}(C)$.

    We first look at the gates lifted from $C$, $\mathfrak{D}(\overline{L}_i)$. From Fig.~\ref{fig:mapping}, it can be seen that the lifted gates are either fold-transversal or SWAP-transversal on $\mathfrak{D}(C)$, as any CNOT gates are applied between qubits that are mapped to each other according to the $ZX$-duality $\tau$, i.e. qubits $i$ and $n+i$. In the $C_4 \otimes_\tau \mathfrak{D}(C)$ code, the corresponding logical qubits are assigned to the same $C_4$ codeblock, and so applying a CNOT between the logical qubits can be accomplished as $\overline{CNOT}_{1,2} = SWAP_{2,4}$ or $\overline{CNOT}_{2,1} = SWAP_{3,4}$. Hence any logical gate consisting of fold-transversal CNOT gates on $\mathfrak{D}(C)$ is SWAP-transversal on $C_4 \otimes_\tau \mathfrak{D}(C)$.
    
    Ensuring that the permutation (automorphism) gates~\cite{Grassl_2013} of $\mathfrak{D}(C)$ remain SWAP-transversal on $C_4 \otimes_\tau \mathfrak{D}(C)$ is more subtle. 
    In particular, we are concerned about SWAP operations that make it so logical qubits $i$ and $n+i$ are in different $C_4$ codeblocks post permutation, potentially removing the ability to perform gates from $G_\tau$ SWAP-transversally. 
    In this scenario, the logical qubits have to be physically swapped between $C_4$ codeblocks to restore the appropriate logical qubit-$C_4$ codeblock assignments, see Appendix C.8 of Ref.~\cite{berthusen2025}. 
    Any SWAP gates on $\mathfrak{D}(C)$ lifted from $C$ respect $\tau$ in the sense that logical qubits $i$ and $n+i$ remain in the same $C_4$ codeblocks after the permutation: SWAP$_{i, n+i}$ keeps the logical qubits within their original $C_4$ codeblocks, and can be implemented on the concatenated code using the SWAP-transversal gate $\overline{\text{SWAP}}_{1,2} = \text{SWAP}_{2,3}$. 
    The lifted version of a SWAP gate, SWAP$_{i,j}$ SWAP$_{i+n, j+n}$, essentially swaps the labels of the $C_4$ codeblocks themselves, which can be tracked in software. Alternatively, the four physical qubits of each codeblock can be swapped. The fact that the same permutation is done across the fold ensures that logical qubits $(i, i+n)$ and $(j, j+n)$ remain in the same $C_4$ codeblocks, respectively. 
    
    We now look at implementing Eq.~\eqref{eq:hswap}, $\overline{H}_\tau = \bigotimes_{i=1}^{2n} H_i ~\bigotimes_{i=1}^n SWAP_{i, n+i}$ on the $C_4 \otimes_\tau \mathfrak{D}(C)$ code. 
    Implementing $\overline{H}_1 \overline{H}_2$ for both $C_4$ code can be done by applying $H_1H_2H_3H_4 \text{SWAP}_{2,3}$. Applying SWAP$_{i, n+i}$ is similarly straightforward: since we have ensured that logical qubits $i$ and $n+i$ are assigned to the same $C_4$ block, we can apply $\overline{\text{SWAP}}_{1,2} = \text{SWAP}_{2,3}$ on each $C_4$ block. Notice that the implementation for $\overline{H}_1 \overline{H}_2$ implicitly has the desired logical SWAP, which we undo. 

    It remains to show that Eq.~\eqref{eq:s}, $\overline{S}_\tau = \bigotimes_{i=1}^n\overline{CZ}_{i,n+i}$ is SWAP-transversal in the concatenated code. Based on the $ZX$-duality $\tau$, we ensured that qubits $i$ and $n+i$ reside in the same $C_4$ block. The $C_4$ code has the intrablock logical gate $\overline{CZ} = S_1^\dagger S_2 S_3 S_4^\dagger$. Applying this gate on all $C_4$ codeblocks enacts the desired logical operation. This result can also be derived from the fact that $C_4 \otimes \mathfrak{D}(C)$ is doubly-even and self-dual, implying that it has this transversal phase gate. However, the logical action is a $CZ$-type circuit, and so a phase gate, Eq.~\eqref{eq:s_tele}, is required for Clifford completeness. 
    Hence, we have showed that $\mathfrak{D}(\overline{L}_i)$, $\overline{H}_{\tau}$, and $\overline{S}_\tau$ are all SWAP-transversal on the $C_4$-CSD codes.
\end{proof}

An interesting consequence of Theorem~\ref{thm:transversal} is the fact that any logical operation from $G_\tau$ can be implemented using only physical single-qubit gates and qubit relabeling.
Fig.~\ref{fig:circuit_simplified} shows an example of a physical implementation of $\overline{H}_\tau$, $\overline{S}_\tau$, and automorphism gates. One can imagine ``untangling'' the SWAP gates in the circuit to obtain an equivalent circuit where the wires do not cross. In this form, it is easy to see that the circuit for each physical qubit is independent and consists of gates from the single-qubit Clifford group $\mathcal{C}_1 = \langle H, S\rangle$. Indeed, regardless of how many single-qubit gates are applied, they can be compressed down into a single Clifford gate $U_i \in \mathcal{C}_1$. In effect, we have shown that an arbitrary length Clifford circuit using only gates from $G_\tau$ can be simplified to a single layer of physical single-qubit Clifford gates  and qubit relabeling, which can be tracked in software. 

Regardless of the choice of seed non-CSS code $C$ and number of logical gates $\overline{L}_i$ lifted from $C$ to $\mathfrak{D}(C)$, the resulting gateset $G_\tau$ is not powerful enough to generate the symplectic group on $2k$ qubits, i.e. $G_\tau \ncong \mathrm{Sp}_{4k}(\mathbb{F}_2)$. 
The reason for this insufficiency is due to the structure of $\overline{S}_\tau$: since $S_{ii} = 0$, we are not able to transform $\overline{X}_i \rightarrow \overline{Y}_i$. To supplement this, we need additional phase gates. 
Some symplectic double codes have logical gates implementing diagonal Clifford operators, see Table~\ref{tab:actions}; however, there is no guarantee that the physical action would lift nicely onto the corresponding $C_4$-CSD code. Instead, we search for phase-type gates directly on the concatenated code. Since $C_4$-CSD codes are self-dual CSS codes, it allows us to make use of the following Lemma, generalized from Theorem 3 of Ref.~\cite{tansuwannont2026}.

\begin{lem}
    Consider a self-dual CSS code $Q$ constructed from classical parity check matrix $H$, i.e., $H_X = H_Z = H$. Let $\pi \in \textrm{Aut}(H)$ be an order-2 automorphism of the underlying classical code. Then the following is a valid logical gate on $Q$:
    \begin{equation}
        \overline{U}_P(\pi) = \bigotimes_{q: \pi(q)=q} S_q  \bigotimes_{\substack{\{q,\,\pi(q)\}\\ \pi(q)\neq q}} CZ_{q, \pi(q) }
    \end{equation}
    \label{lem:fold}
\end{lem}

\begin{proof}
    We must show that $\overline{U}_P(\pi)$ preserves the stabilizer group $\mathcal{S}$ of $Q$ and acts as a valid logical Clifford gate on its logical qubits. Since $Q$ is a self-dual CSS code, its stabilizer group is generated by $X$-type stabilizers $X^{\otimes s}$ and $Z$-type stabilizers $Z^{\otimes s}$ for $s \in \textrm{row}(H)$, and $\mathcal{S}$ contains all products $X^{\otimes s_1} Z^{\otimes s_2}$ for $s_1, s_2 \in \textrm{row}(H)$.

    Let $P_\pi$ denote the $n\times n$ permutation matrix of $\pi$. Since $\pi^2 = \mathrm{id}$, we have $P_\pi^{-1} = P_\pi$, and since $P_\pi^{-1} = P_\pi^\top$ for any permutation matrix, it follows that $P_\pi = P_\pi^\top$. One can verify that conjugation by $\overline{U}_P(\pi)$ acts on $X^{\otimes a} Z^{\otimes b}$ as $[a~|~b] \rightarrow [a~|~b + P_\pi a]$ in the symplectic representation; this map is symplectic precisely because $P_\pi = P_\pi^\top$. Hence under $\overline{U}_P(\pi)$, a $Z$-type stabilizer $[0~|~s]$ maps to $[0~|~s]$ and is trivially preserved. An $X$-type stabilizer $[s~|~0]$ maps to $[s~|~P_\pi s]$. Since $\pi \in \mathrm{Aut}(C)$, it preserves $C^\perp = \mathrm{row}(H)$, and so $P_\pi s \in \mathrm{row}(H)$. The image $X^{\otimes s} Z^{\otimes P_\pi s}$ then has both components in $\mathrm{row}(H)$ and is therefore in $\mathcal{S}$. We conclude that $\overline{U}_P(\pi)$ preserves the stabilizer group.

    We now determine the logical action. Let $\overline{X}_i = X^{\otimes \ell_i^X}$ and $\overline{Z}_i = Z^{\otimes \ell_i^Z}$ be a choice of logical representatives for $i = 0, \ldots, k-1$, where $\ell_i^X, \ell_i^Z \in \ker(H) \char`\\ \mathrm{row}(H)$ satisfy the condition $\ell_i^X \ell_j^Z = \delta_{ij}$. Since $\pi \in \mathrm{Aut}(C) \simeq \mathrm{Aut}(C^\perp)$~\cite{berthusen2025_2}, it also preserves $C = \ker(H)$ and therefore induces a well-defined map on the quotient $\ker(H)/\mathrm{row}(H)$. Under $\overline{U}_P(\pi)$, the logical $Z$ operators remain unchanged: $[0~|~\ell_i^Z] \to [0~|~\ell_i^Z]$. The $X$ logical operators transform as $[\ell_i^X~|~0] \to [\ell_i^X~|~P_\pi \ell_i^X]$. Since $P_\pi \ell_i^X \in \ker(H)$, the operator $Z^{\otimes P_\pi \ell_i^X}$ is a product of logical $\overline{Z}$ operators (modulo stabilizers). To determine which logical $\overline{Z}_j$ appear, we check commutation: $Z^{\otimes P_\pi \ell_i^X}$ anticommutes with $\overline{X}_j = X^{\otimes \ell_j^X}$ if and only if $(\ell_j^X)^\top P_\pi \ell_i^X = 1$, indicating the presence of $\overline{Z}_j$. Defining $\Gamma_{ji} = (\ell_j^X)^\top P_\pi \ell_i^X$ for $i, j \in [k]$, the full logical action is:

    \begin{equation}
        \overline{X}_i \rightarrow \overline{X}_i \prod_j \overline{Z}_j^{\Gamma_{ji}}, \qquad \overline{Z}_i \rightarrow \overline{Z}_i.
    \end{equation}
    In the $k$-qubit logical symplectic representation this performs the map $[a~|~b] \rightarrow [a~|~b + \Gamma a]$. The matrix $\Gamma$ is symmetric, as follows directly from $P_\pi = P_\pi^\top$,
    \begin{equation}
        \Gamma_{ji} = (\ell_j^X)^\top P_\pi\, \ell_i^X = (\ell_j^X)^\top P_\pi^\top\, \ell_i^X = (P_\pi\, \ell_j^X)^\top \ell_i^X = (\ell_i^X)^\top P_\pi\, \ell_j^X = \Gamma_{ij}.
    \end{equation}
    Hence $\overline{U}_P(\pi)$ decomposes as $\overline{S}$ on logical qubit $i$ whenever $\Gamma_{ii} = 1$ and $\overline{\mathrm{CZ}}$ between logical qubits $i$ and $j$ whenever $\Gamma_{ij} = 1$ for $i \neq j$.
\end{proof}

These gates are fold-transversal gates in the terminology of Ref.~\cite{Breuckmann_2024}. These gates are not `free' like those in $G_\tau$; while depth-1, they require physical two-qubit gates, and additional gadgetry to be fault-tolerant, see Section~\ref{sec:y_prep}. Note that finding gates in this way also has its limitations. In particular, the cost of computing the automorphism group of a classical code grows exponentially in the code dimension~\cite{leon1982}. One optimization, which we apply here, is to instead compute matrix automorphisms of the underlying classical code~\cite{sayginel2024fault}. Explicitly, we construct a matrix $M$ whose rows are an overcomplete set of parity checks from $H$. We then map $M$ to a bipartite graph and use the open-source library Bliss~\cite{JunttilaKaski} to compute its graph automorphisms, which correspond to qubit permutations that preserve the stabilizer group of the $C_4$-CSD code. From the automorphism group generators, we enumerate involutions via random word search and conjugation closure, then lift each involution $\pi$ to a depth-1 CZ-S circuit $\overline{U}_P(\pi)$.

All $\overline{U}_P(\pi)$ commute since they lie in the diagonal Clifford group, and as such $\langle \overline{U}_P(\pi)\rangle$ is an elementary abelian 2-group $\mathbb{Z}_2^r$ of size $2^r$, where $r$ is the number of linearly independent $\overline{U}_P(\pi)$'s. The maximum achievable size is $| \langle\overline{U}_P(\pi) \rangle| \le 2^{2k\cdot(2k+1)/2}$, corresponding to the full diagonal Clifford group on $2k$ qubits. We present the size of $\langle \overline{U}_P(\pi)\rangle$ for the codes in Table~\ref{tab:codes}. We indicate where $\langle \overline{U}_P(\pi)\rangle$ attains the maximal size, and also where we only obtain a lower bound on the size due to not exhaustively enumerating all involutions of the underlying classical code. 

With these additional gates, we now have an expanded logical gateset $G = \langle\mathfrak{D}
'(\overline{L}_i), \overline{H}_\tau$, $\overline{S}_\tau, \overline{U}_P(\pi)\rangle$ which, for certain choices of $C$, is sufficiently powerful to generate the symplectic group on all $2k$ qubits, $G \cong \mathrm{Sp}_{4k}(\mathbb{F}_2)$. Codes in Table~\ref{tab:codes} with a green checkmark are verified using GAP~\cite{GAP4} to have this property. In general, knowing apriori whether $G \cong \mathrm{Sp}_{4k}(\mathbb{F}_2)$ is difficult, and brute-force checking is the surest way of determining the power of $G$. Of course, enumerating the full group is infeasible for $k > 8$, and so smarter ways of checking for completeness could be used, such as computing orbits of Pauli exponentials.

While $\overline{U}_P(\pi)$ could be implemented in-place, we choose to apply them through a state injection scheme~\cite{Zhou2000} as depicted in the following circuit or similar:
\begin{equation}
\begin{quantikz}[background color=white, row sep={1cm,between origins}]
\lstick{$\overline{U}_P(\pi) \ket{\overline{+}}^{\otimes 2k}$} & \qwbundle{2n} & \ctrl{1} &  \gate{\overline{U}_P(\pi) X^{\otimes 2n} \overline{U}_P(\pi)^\dagger} \wire[d][1]{c} & \rstick{$\overline{U}_P(\pi) \ket{\overline\psi}$}\\
\lstick{$\ket{\overline\psi}$} & \qwbundle{2n} & \targ{}  &  \meter{}           
\end{quantikz} \approx 
\begin{quantikz}[background color=white, row sep={1cm,between origins}]
\lstick{$\ket{\overline{\psi}}$} & \gate[2]{\overline{U}_P(\pi)} & \rstick{$\overline{U}_P(\pi) \ket{\overline\psi}$} \\
& 
\end{quantikz}
\label{eq:s_tele}
\end{equation}

Performing gates in this way requires additional overhead in the form of ancilla blocks prepared in the $\overline{U}_P(\pi)\ket{\overline{+}}$ state; however, we have several reasons for this choice, also see Section~\ref{sec:qec}. One reason is for instances where $G \ncong \mathrm{Sp}_{4k}(\mathbb{F}_2)$. We are able to increase the power of $\langle \overline{U}_P(\pi) \rangle$ by instead applying $\overline{U}_P(\pi)$ to a partial plus state~\cite{Goto_2024} before injecting. Doing so allows us to selectively activate subsets of the logical $S$ and CZ gates contained in $\overline{U}_P(\pi)$. Concretely, if we prepare only a subset $s \subset \{1, ..., 2k\}$, of the logical qubits in $\ket{\overline{+}}$ and the remainder in $\ket{\overline{0}}$, then both the $S$ and CZ gates act trivially on the qubits in the $\ket{\overline{0}}$ state. The resulting ancilla state therefore encodes only a restriction of $\overline{U}_P(\pi)$ to the qubits indexed by $s$, which may not lie in $\langle \overline{U}_P(\pi) \rangle$. When we include the resulting gates $\overline{U}_{P}(\pi)_{| s}$ from by varying $s$ over all $2^{2k}$ subsets, it may be the case that the new full gateset $G' = \langle\mathfrak{D}
'(\overline{L}_i), \overline{H}_\tau, \overline{S}_\tau, \overline{U}_{P}(\pi)_{| s}\rangle \cong \mathrm{Sp}_{4k}(\mathbb{F}_2)$. Such is true for the $[[48,4,8]]$ $C_4$-CSD code. Unfortunately, even this expanded gateset is not sufficient for some codes to be Clifford universal---in particular, when $\overline{U}_{P}(\pi)_{| s}$ has no gates that transform $\overline{X}_i \rightarrow \overline{Y}_i$; for these codes, indicated by a red $X$ in Table~\ref{tab:codes}, the Clifford group can be completed by injecting logical phase gates using, for example, universal adapters~\cite{swaroop2025}. Given the flexibility we have in choosing the seed code $C$, we expect there to be $C_4$-CSD codes for arbitrary $[[n,k,d]]$ on which $G$ or $G'$ generates the full Clifford group, without the need for complex adapters.

Whether or not we inject $\overline{U}_P(\pi)$ or apply it in place, automorphism gates from $G_\tau$ and fold-transversal gates from $\langle \overline{U}_P(\pi) \rangle$ (or $\langle \overline{U}_{P}(\pi)_{| s} \rangle$)can be combined in a way that yields conceptually simple physical implementations of logical Clifford circuits. In particular, an arbitrary $2k$-logical qubit Clifford operator $\overline{U}$ will have some decomposition into $G_\tau$ and $\overline{U}_P(\pi)$ gates. Since any element of $G_\tau$ can be compiled into a single layer of single-qubit gates, the resulting physical circuit simply alternates between $\overline{U}_P(\pi)$ gates and applying a layer of single-qubit gates. Below we show the resulting form of a general logical Clifford circuit when using Eq.~\eqref{eq:s_tele} to inject $\overline{U}_P(\pi)$ gates:

\begin{equation}
\begin{quantikz}[background color=white, row sep={1cm,between origins}]
\lstick{$\ket{\overline{\psi}}$} & \gate{U_i^{(1)}} & \gate[2]{\overline{U}_P(\pi)} & \gate{U_i^{(2)}} & \gate[2]{\overline{U}_P(\pi)} & \gate{U_i^{(3)}} & \\
&\wireoverride{n} & & \wireoverride{n} & & \wireoverride{n} & \wireoverride{n}   
\end{quantikz} ...
\begin{quantikz}[background color=white, row sep={1cm,between origins}]
& \gate[2]{\overline{U}_P(\pi)} & \gate{U_i^{(m)}} & \rstick{$\overline{U} \ket{\overline\psi}$} \\
& &\wireoverride{n}    
\end{quantikz}
\label{eq:s_tele_circuit}
\end{equation}
Moreover, the final step of state injection is a Pauli correction. The single-qubit gates could be folded into this Pauli correction to obtain a physical circuit that essentially only consists of logical $\overline{U}_P(\pi)$ gate injections.

While we have the full Clifford group on a single $C_4 \otimes_\tau \mathfrak{D}(C)$ codeblock, large algorithms will require logical quantum circuits interacting logical qubits across many codeblocks. Here we show that the full Clifford group on $b$ codeblocks each containing $2k$ qubits is achievable using transversal gates and injected $\overline{U}_P(\pi)$ gates. Being CSS codes, CSD codes have transversal CNOT gates $\bigotimes_{i=1}^{2k} \overline{CNOT}_{i,2k+i} = \bigotimes_{i=1}^{4n} CNOT_{i, 4n+i}$ (and another swapping the controls and targets). Hence using gates from $G$ on each codeblock in addition to the transversal CNOT gates between blocks yields the following group of achievable gates: 

\renewcommand\arraystretch{1}
\begin{equation}
G^{\otimes 2} = \Biggr \langle
\begin{pmatrix}
\begin{array}{cc|cc}
I_k & I_k & 0 & 0 \\
0 & I_k & 0 & 0 \\
\hline
0 & 0 & I_k & 0 \\
0 & 0 & I_k & I_k
\end{array}
\end{pmatrix},
\begin{pmatrix}
\begin{array}{cc|cc}
I_k & 0 & 0 & 0 \\
I_k & I_k & 0 & 0 \\
\hline
0 & 0 & I_k & I_k \\
0 & 0 & 0 & I_k
\end{array}
\end{pmatrix},
\begin{pmatrix}
\begin{array}{cc|cc}
A_1 & 0 & B_1 & 0 \\
0 & 0 & 0 & 0 \\
\hline
C_1 & 0 & D_1 & 0 \\
0 & 0 & 0 & 0
\end{array}
\end{pmatrix},
\begin{pmatrix}
\begin{array}{cc|cc}
0 & 0 & 0 & 0 \\
0 & A_2 & 0 & B_2 \\
\hline
0 & 0 & 0 & 0 \\
0 & C_2 & 0 & D_2 
\end{array}
\end{pmatrix} \Biggr \rangle
\end{equation}
where
\begin{equation}
\begin{pmatrix}
\begin{array}{c|c}
A_1 & B_1 \\
\hline
C_1 & D_1 
\end{array}
\end{pmatrix},
\begin{pmatrix}
\begin{array}{c|c}
A_2 & B_2 \\
\hline
C_2 & D_2 
\end{array}
\end{pmatrix} \in \mathrm{Sp}_{2k}(\mathbb{F}_2).
\end{equation}

\begin{thm}
    If $G \cong \textrm{Sp}_{4k}(\mathbb{F}_2)$, then $G^{\otimes 2} \cong \textrm{Sp}_{8k}(\mathbb{F}_2)$.
    \label{thm:two_block}
\end{thm}

\begin{proof}
    Let us first consider Clifford circuits containing only CNOT gates. Instead of a full symplectic matrix in $\textrm{Sp}_{2k}(\mathbb{F}_2)$, we can describe a CNOT circuit with only the upper left $k\times k$ block $M$, see Eq.~\eqref{eq:gates}. This is because for a circuit consisting of only CNOT gates, the off-diagonal blocks of the corresponding symplectic matrix will be zero, and the bottom right block can be obtained as $(M^{-1})^\top$. Now consider the following $2k \times 2k$ matrix that can be obtained using $A \in \mathrm{GL}_{k}(\mathbb{F}_2)$ and transversal CNOT gates.
    \begin{equation}
        \begin{pmatrix}
        \begin{array}{c|c}
        A & 0 \\
        \hline
        0 & 0 
        \end{array}
        \end{pmatrix}
        \begin{pmatrix}
        \begin{array}{c|c}
        I & I \\
        \hline
        0 & I 
        \end{array}
        \end{pmatrix}
        \begin{pmatrix}
        \begin{array}{c|c}
        A & 0 \\
        \hline
        0 & 0 
        \end{array}
        \end{pmatrix}^{-1} = 
        \begin{pmatrix}
        \begin{array}{c|c}
        I & A \\
        \hline
        0 & I 
        \end{array}
        \end{pmatrix}
        \label{eq:off_diag}
    \end{equation}
    The corresponding lower triangular matrix can be obtained by instead using the transversal CNOT gate with the control and targets switched. An important property to note about matrices of the form of Eq.~\eqref{eq:off_diag} is that composing matrices has the effect of adding the off-diagonal blocks over $\mathbb{F}_2$. In particular, for $A, B \in \mathrm{GL}_{k}(\mathbb{F}_2)$,
    \begin{equation}
        \begin{pmatrix}
        \begin{array}{c|c}
        I & A \\
        \hline
        0 & I 
        \end{array}
        \end{pmatrix}
        \begin{pmatrix}
        \begin{array}{c|c}
        I & B \\
        \hline
        0 & I 
        \end{array}
        \end{pmatrix} = 
        \begin{pmatrix}
        \begin{array}{c|c}
        I & A + B \\
        \hline
        0 & I 
        \end{array}
        \end{pmatrix}.
        \label{eq:mat_product}
    \end{equation}
    We now make use of the following two results:

    \begin{thm}[Theorem 4 of~\cite{Grassl_2013}]
        Let $A, B \in M_k(\mathbb{F}_2)$ be arbitrary binary $k \times k$ matrices, and consider the group $G_{12}$ generated by
        \begin{equation}
            G_{12} = \Biggl \langle
            \begin{pmatrix}
            \begin{array}{c|c}
            I & A \\
            \hline
            0 & I 
            \end{array}
            \end{pmatrix},
            \begin{pmatrix}
            \begin{array}{c|c}
            I & 0 \\
            \hline
            B & I 
            \end{array}
            \end{pmatrix} \Biggr \rangle.
        \end{equation}
        Then $G_{12} \cong \mathrm{GL}_{2k}(\mathbb{F}_2)$.
        \label{thm:auts}
    \end{thm}
    
    \begin{thm}[Corollary IX.8 of~\cite{malcolm2025}]
        Let $A \in M_k(\mathbb{F}_2)$. Then there exist $w$ pairs of invertible matrices $g_{i_1}$, $g_{i_2} \in \mathrm{GL}_k(\mathbb{F}_2)$, for some $w \in O(k)$ such that $A = \sum_{i_1, i_2}^w g_{i_1} \otimes g_{i_2}$.
        \label{thm:shyps}
    \end{thm}
    Since we have access to the full Clifford group on each codeblock, and $\mathrm{GL}_{2k}(\mathbb{F}_2) \subset \mathrm{Sp}_{4k}(\mathbb{F}_2)$, Theorem~\ref{thm:shyps} tells us that we can achieve arbitrary $A \in M_{2k}(\mathbb{F}_2)$ in the off-diagonal blocks of Eq.~\eqref{eq:off_diag}. As such, we can then apply Theorem~\ref{thm:auts} which implies that all CNOT gates on and between both codeblocks, i.e. $\mathrm{GL}_{4k}(\mathbb{F}_2)$, are achievable using only gates from $G^{\otimes 2}$. 
    It is well known that the Clifford group is generated by CNOT gates and single-qubit Clifford gates, and so we have the immediate result that the entire Clifford group on all $4k$ logical qubits is implementable using only gates from $G^{\otimes 2}$, that is $G^{\otimes 2} \cong \mathrm{Sp}_{8k}(\mathbb{F}_2)$.
\end{proof}

We can generalize this to $b$ blocks by considering all $2{b \choose 2}$ possible transversal CNOT pairings as well as $G$ on each codeblock. Denote as $G^{\otimes b}$ the resulting group generated by all of these gates. Then we have the following corollary as a result of Thm.~\ref{thm:two_block}: 
\begin{cor}
    If $G \cong \textrm{Sp}_{4k}(\mathbb{F}_2)$, then $G^{\otimes b} \cong \textrm{Sp}_{4kb}(\mathbb{F}_2)$.
\end{cor}

\begin{proof}
    From Thm.~\ref{thm:two_block}, we can construct arbitrary CNOT gates between any of the $b \choose 2$ block pairings. Thus we have access to any gate in $\mathrm{GL}_{2bk}(\mathbb{F}_2)$. The full Clifford group $\mathrm{Sp}_{4bk}(\mathbb{F}_2)$ is generated by these CNOT gates and single-qubit Cliffords on each logical qubit.
\end{proof}

\subsubsection{$[[16,4,4]]$ code}

\newcolumntype{C}[1]{>{\centering\arraybackslash}m{#1}}
\begin{table}
\centering
\begin{tabular}{|C{1.75cm}    
    |C{1.75cm}   
    |C{2.5cm}     
    |C{2cm}   
    |C{3cm}   
    |C{1cm}|}
\hline
\multicolumn{2}{|c|}{$C$ $[[4,2,2]]$} & 
\multicolumn{2}{c|}{$\mathfrak{D}(C)$ $[[8,4,2]]$} & 
\multicolumn{2}{c|}{$C_4 \otimes_\tau \mathfrak{D}(C)$ $[[16,4,4]]$} \\
\hline
Phys. imp. & Log. imp. & Phys. imp. & Log. imp. & Phys. imp. & Log. imp. \\
\hline
\begin{tikzpicture}
\node[scale=0.6] {
\begin{quantikz}[background color=white, row sep={0.75cm,between origins}]
& & & \\
& & & \\
& \gate{H} & \swap{1} & \\
& \gate{H} & \targX{} &
\end{quantikz} };
\end{tikzpicture} & 
\begin{tikzpicture}
\node[scale=0.6] {
\begin{quantikz}[background color=white, row sep={0.75cm,between origins}]
& \gate{H} & \swap{1} & \\
& \gate{H} & \targX{} &
\end{quantikz} };
\end{tikzpicture}
& \raggedright \scriptsize $SWAP_{3,7}$ \newline $SWAP_{4,8}$ \newline $SWAP_{3,4}$ \newline $SWAP_{7,8}$ & 
\begin{tikzpicture}
\node[scale=0.6] {
\begin{quantikz}[background color=white, row sep={0.75cm,between origins}]
& \ctrl{3} & & \\
& & \ctrl{1} & \\
& & \targ{} & \\
& \targ{} & &
\end{quantikz} };
\end{tikzpicture} & \raggedright \scriptsize $\bigotimes_{i=2}^3 SWAP_{2+4i, 3+4i}$ \newline $\bigotimes_{i=1}^4 SWAP_{i+8, i+12}$ & $\curvearrowleft$ \\
\hline
\begin{tikzpicture}
\node[scale=0.6] {
\begin{quantikz}[background color=white, row sep={0.75cm,between origins}]
& \gate{H} &  \\
& \gate{H} & \\
& \gate{H} & \\
& \gate{H} &
\end{quantikz} };
\end{tikzpicture} & 
\begin{tikzpicture}
\node[scale=0.6] {
\begin{quantikz}[background color=white, row sep={0.75cm,between origins}]
& \swap{1} & \\
& \targX{} &
\end{quantikz} };
\end{tikzpicture}
& \raggedright \scriptsize $\bigotimes_{i=1}^4 SWAP_{i,i+4}$ & 
\begin{tikzpicture}
\node[scale=0.6] {
\begin{quantikz}[background color=white, row sep={0.75cm,between origins}]
& \swap{1} & \\
& \targX{} & \\
& \swap{1} & \\
& \targX{} &
\end{quantikz} };
\end{tikzpicture} & \raggedright \scriptsize $\bigotimes_{i=0}^3 SWAP_{2+4i, 3+4i}$ & $\curvearrowleft$ \\
\hline
\begin{tikzpicture}
\node[scale=0.6] {
\begin{quantikz}[background color=white, row sep={0.75cm,between origins}]
& & & & \\
& \gate{H} & & \swap{2} &\\
& & \swap{1} & & \\
& \gate{H} & \targX{} & \targX{} & 
\end{quantikz} };
\end{tikzpicture} & 
\raggedright\begin{tikzpicture}
\node[scale=0.5] {
\begin{quantikz}[background color=white, row sep={0.75cm,between origins}]
& \gate{H} & \swap{1} & \gate[2]{U} &\\
& \gate{H} & \targX{} & &
\end{quantikz} };
\end{tikzpicture}
& \raggedright\scriptsize $SWAP_{2, 6}$ \newline $SWAP_{4, 8}$ \newline $SWAP_{3, 4}$ \newline $SWAP_{7, 8}$ \newline $SWAP_{2, 4}$ \newline $SWAP_{6, 8}$ & 
\begin{tikzpicture}
\node[scale=0.5] {
\begin{quantikz}[background color=white, row sep={0.75cm,between origins}]
& \swap{3} & \ctrl{3} & & & \\
& & & \swap{1} & \ctrl{1} & \\
& & & \targX{} & \targ{} & \\
& \targX{} & \targ{} & & &
\end{quantikz} };
\end{tikzpicture}
& \raggedright \scriptsize $SWAP_{6, 7}$ \newline $SWAP_{14, 15}$ \newline $\bigotimes_{i=1}^4 SWAP_{i+8, i+12}$ \newline $\bigotimes_{i=1}^4 SWAP_{i+4, i+12}$ & $\curvearrowleft$ \\
\hline
\begin{tikzpicture}
\node[scale=0.6] {
\begin{quantikz}[background color=white, row sep={0.75cm,between origins}]
& \gate{H} & \gate{S} &  \\
& \gate{S^\dagger} & \gate{H} & \\
& \gate{S^\dagger} & \gate{H} & \\
& \gate{H} & \gate{S} &
\end{quantikz} };
\end{tikzpicture}
& 
\begin{tikzpicture}
\node[scale=0.6] {
\begin{quantikz}[background color=white, row sep={0.75cm,between origins}]
& \ctrl{1} & \swap{1} & \\
& \targ{} & \targX{} &
\end{quantikz} };
\end{tikzpicture}
& \raggedright \scriptsize $SWAP_{1, 5} CNOT_{1,5}$ \newline $CNOT_{2,6} SWAP_{2,6}$ \newline $CNOT_{3,7} SWAP_{3,7}$ \newline  $SWAP_{4,8} CNOT_{4,8}$& 
\begin{tikzpicture}
\node[scale=0.6] {
\begin{quantikz}[background color=white, row sep={0.75cm,between origins}]
& \swap{1} & \ctrl{1} & \\
& \targX{} & \targ{} & \\
& \swap{1} & \targ{} & \\
& \targX{} & \ctrl{-1} &
\end{quantikz} };
\end{tikzpicture}& \raggedright \scriptsize $SWAP_{2,3} SWAP_{2,4}$ \newline $SWAP_{6,8} SWAP_{6,7}$ \newline $SWAP_{10,12} SWAP_{10,11}$ \newline  $SWAP_{14,15} SWAP_{14,16}$ & $\curvearrowleft$ \\
\hline
\rowcolor[RGB]{230, 242, 255} - & - & \raggedright \scriptsize $\bigotimes_{i=1}^8 H_i$ \newline $\bigotimes_{i=1}^4 SWAP_{i, i+4}$ & 
\begin{tikzpicture}
\node[scale=0.6] {
\begin{quantikz}[background color=white, row sep={0.75cm,between origins}]
& \gate{H} & \swap{2} & &  \\
& \gate{H} & & \swap{2} & \\
& \gate{H} & \targX{} & & \\
& \gate{H} & & \targX{} &
\end{quantikz} };
\end{tikzpicture}
& \raggedright \scriptsize $\bigotimes_{i=1}^{16} H_i$ & $\curvearrowleft$ \\
\hline
\rowcolor[RGB]{230, 242, 255} - & - & \raggedright \scriptsize $\bigotimes_{i=1}^4 CZ_{i, i+4}$ & 
\begin{tikzpicture}
\node[scale=0.6] {
\begin{quantikz}[background color=white, row sep={0.75cm,between origins}]
& \ctrl{2} & &  \\
& & \ctrl{2} & \\
& \control{} & & \\
& & \control{} &
\end{quantikz} };
\end{tikzpicture}
& \raggedright \scriptsize $\bigotimes_{i=0}^3 S^\dagger_{1+4i} S^\dagger_{4+4i}$ \newline $\bigotimes_{i=0}^3 S_{2+4i} S_{3+4i}$  & $\curvearrowleft$ \\
\hline
\rowcolor[RGB]{255, 222, 222} - & - & \raggedright\scriptsize $\bigotimes_{i=1}^8 S_i$ & 
\begin{tikzpicture}
\node[scale=0.6] {
\begin{quantikz}[background color=white, row sep={0.75cm,between origins}]
& \ctrl{3} & &  \\
& & \ctrl{1} & \\
& & \control{} & \\
& \control{} & &
\end{quantikz} };
\end{tikzpicture}
& \raggedright\scriptsize Ref.~\cite{berthusen2025} Appendix D.7 & $\curvearrowleft$ \\
\hline
\rowcolor[RGB]{255, 222, 222} - & - & \raggedright\scriptsize $SWAP_{6,7}$ \newline $SWAP_{3,5}$ & 
\begin{tikzpicture}
\node[scale=0.6] {
\begin{quantikz}[background color=white, row sep={0.75cm,between origins}]
&  & \targ{} & \\
& \swap{1} & & \\
& \targX{} & & \\
&  & \ctrl{-3} &
\end{quantikz} };
\end{tikzpicture}
& \raggedright\scriptsize Ref.~\cite{berthusen2025} Appendix D.8 & $\curvearrowleft$ \\
\hline 
\end{tabular}
\caption{\small SWAP-transversal gates on the non-CSS $[[4,2,2]]$ code lifting to SWAP-transversal logical gates on the $[[8,4,2]]$ symplectic double code and $[[16,4,4]]$ concatenated symplectic double code. The two-qubit gate $U$ in the third row is the $C(X,X)$ gate~\cite{sayginel2024fault}. The two rows highlighted blue are gates arising from the double cover, $\overline{H}_\tau$ and $\overline{S}_\tau$. The rows highlighted red are gates that require additional overhead to implement fault-tolerantly on the $[[16,4,4]]$ code, see Appendix D.7 of Ref.~\cite{berthusen2025}. While the physical actions are different, the logical actions on the $[[8,4,2]]$ symplectic double code and the $[[16,4,4]]$ concatenated symplectic double code are the same.}
\label{tab:actions}
\end{table}

Here we look at an explicit example of a $C_4$-CSD code, namely the $[[16,4,4]]$ code, and list the available logical gates for each code in the construction pipeline. The base non-CSS for this code is a $[[4,2,2]]$ code, which can be interpreted as a Hadamard-transformed~\cite{Bonilla_Ataides_2021} version of the $C_4$ code defined in Section~\ref{sec:code_construction}. Specifically, we define this $[[4,2,2]]$ code with two stabilizer generators $XZZX, ZXXZ$ and two pairs of anticommuting logical operators $IZZI, ZIXI, ZIIZ, IZIX$. Similar to the CSS version, this $[[4,2,2]]$ code has several SWAP-transversal gates. Valid automorphisms and the corresponding logical actions are identified using the \texttt{autqec} package~\cite{autqec_code} and displayed in the first two columns of Table~\ref{tab:actions}. Using this code, we construct a symplectic double code by applying Eq.~\eqref{eq:pcm_double}. The parity check matrix in symplectic form for the resulting $[[8,4,2]]$ code is then:
\begin{align}
H_{[[8,4,2]]} =  \left(\begin{array}{c|c}
10010110 & 00000000 \\
01101001 & 00000000 \\
00000000 & 01101001 \\
00000000 & 10010110
\end{array} \right)
\end{align}

The physical implementation of the SWAP-transversal gates of the $[[4,2,2]]$ code lifts according to $\mathfrak{D}'$ to yield logical gates on the symplectic double $[[8,4,2]]$ code. This mapping is shown in Fig.~\ref{fig:mapping}, and can be verified between the first and third columns of Table~\ref{tab:actions}. The resulting logical action was then computed for each physical circuit. In addition to lifting the SWAP-transversal gates, we gain additional gates, $\overline{H}_\tau$ and $\overline{S}_\tau$, arising from the double cover. These two gates are listed in Table~\ref{tab:actions} in the blue highlighted rows. As mentioned in Section~\ref{sec:log_clifford}, every symplectic double code will have these two additional gates. The $[[8,4,2]]$ code has additional SWAP-transversal gates that do not arise from the lifted $[[4,2,2]]$ gates, one of which arises due to an automorphism that does not respect $\tau$, while the other is due to the fact that the $[[8,4,2]]$ code is self dual and has a transversal phase gate. This latter gate can also be found directly on the $[[16,4,4]]$ code by searching for fold-transversal gates, Lemma~\ref{lem:fold}.

We then obtain the $C_4$-concatenated symplectic double $[[16,4,4]]$ code by concatenating the $[[8,4,2]]$ code with the $C_4$ code along the ZX-duality $\tau$. For completeness, the parity check matrix for both the $X$- and $Z$-type stabilizer generators is given as follows:
\begin{align}
H_{X, [[16,4,4]]} = H_{Z, [[16,4,4]]} = \left(\begin{array}{c}
1111000000000000 \\
0000111100000000 \\
0000000011110000 \\
0000000000001111 \\
1100101010101100 \\
1010110011001010
\end{array} \right).
\end{align}
Each of the physical gates on the $[[8,4,2]]$ code are converted to `logical' gates acting on the $C_4$ logical qubits according to Theorem~\ref{thm:transversal}. As noted in the theorem, all gates in $G_\tau$ are SWAP-transversal on the concatenated code. While the physical circuits are different between the symplectic double and its concatenated version, the corresponding logical actions are the same. On the $[[16,4,4]]$ code, we see two exceptions where the logical gate on the symplectic double code does not lift to a SWAP-transversal gate on the concatenated code: the transversal phase gate requires performing logical global phase gates on each $C_4$ block; this can be accomplished fault-tolerantly using the circuit in Appendix D.7 of Ref.~\cite{berthusen2025} or Section~\ref{sec:y_prep} but requires two-qubit gates. Additionally, the automorphism gate in the bottom row does not respect $\tau$ and requires swapping logical qubits \textit{between} $C_4$ codeblocks. This is possible using additional $C_4$ ancilla codeblocks, see Appendix D.8 of Ref.~\cite{berthusen2025}. As noted in Table~\ref{tab:codes}, the SWAP-transversal gates for this code yield 216 unique logical gates. Adding in the fold-transversal gates $\overline{U}_P(\pi)$ yields the full symplectic group on the four logical qubits. 

\subsubsection{Compilation}
\label{sec:compilation}

Compilation for $C_4$-CSD codes is conceptually simple in the sense that since $G \cong \textrm{Sp}_{4k}(\mathbb{F}_2)$, any logical Clifford circuit can be expressed using the generators from $G$ and Pauli corrections. We can estimate that an element of $\textrm{Sp}_{4k}(\mathbb{F}_2)$ can be decomposed into a sequence of $O(\log N)$ generators, where $N$ is the order of the symplectic group,  $N = |\textrm{Sp}_{2k}(\mathbb{F}_2)|$ = $2^{k^2} \prod_{i=1}^k 2^{2i} - 1$. This estimate comes from the naive information-theoretic argument that roughly $O(\log N)$ bits of information are required to specify an arbitrary element of the group. For CSD codes, this estimate would imply that any logical Clifford circuit can be decomposed into $O(k^2)$ logical generators. However, even for a reasonable number of logical qubits $|\textrm{Sp}_{4k}(\mathbb{F})_2|$ becomes intractably large, and a brute-force search for a minimum length decomposition becomes intractable. As such, efficient heuristic algorithms that provide reasonable length decompositions will likely be used instead of exhaustive search.
Doing computation with $C_4$-CSD codes would likely be more convenient if the $U_i$'s of individual logical gates, i.e. $\overline{H}_1$ or $\overline{CNOT}_{1,2}$, were precompiled. Then, an $m$ gate logical Clifford circuit could be implemented with an $O(m\log N) \in O(m)$ depth circuit of single-qubit Cliffords and gate injections.   
While no longer constant-depth, this method avoids the complexity of decomposing an arbitrary logical Clifford circuit into the available generators at compile-time.

For the smallest instances of $C_4$-CSD codes, such as the $[[20,2,6]]$ code, it is straightforward and computationally reasonable to perform breadth first search on the Cayley graph of $G$ in order to find the decomposition of minimum length. GAP can be used to find such factorizations, from which we report that every element in $\textrm{Sp}_4(\mathbb{F}_2)$ requires at most three applications of the injected phase gate, with a majority requiring two or fewer. Hence any Clifford circuit can be implemented with at most three state injections. A worse-case example of compiling a circuit into the available logical generators is shown below:
\begin{figure}[H]
    \centering
    \includegraphics[width=0.85\linewidth]{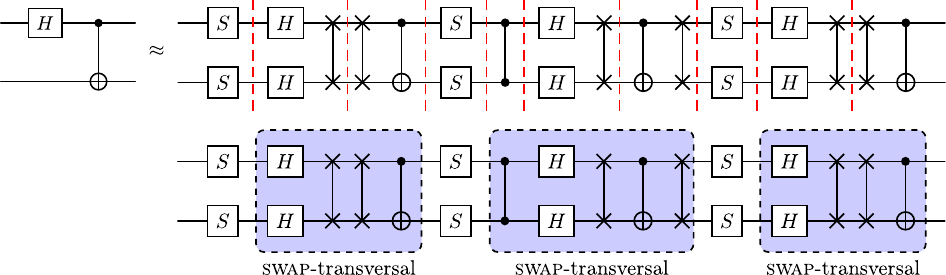}
    \label{fig:2026compilation}
\end{figure}
Each gate sequence between the red dashed lines is a single logical gadget of the $[[20,2,6]]$ $C_4$-CSD code, and when composed in this way implements a logical Bell pair, up to Pauli corrections and a global phase. The logical gadgets in the blue boxes can be implemented SWAP-transversally; consequently, their physical implementation can be compressed to a single layer of physical single-qubit Clifford gates and qubit relabeling, yielding the simple physical circuit shown in Eq.~\eqref{eq:s_tele_circuit}. 
We can also consider the alternative compilation scheme of Section~\ref{sec:qec} in which we inject both $\overline{S}$ and $\overline{\sqrt{X}}$ gates. Verifying the factorization again yields that at most three (combined) instances of state injection were required to implement an arbitrary Clifford circuit on the two logical qubits of the $[[20,2,6]]$ code. And again, a majority of matrices require only one or two $\overline{U}_P(\pi)$ gates, meaning that a majority of Clifford circuits can be implemented in a single QEC cycle. See Fig.~\ref{fig:1644compilation} for an example of this compilation strategy on the $[[16,4,4]]$ $C_4$-CSD code for preparing a graph state. Note, however, that it is not the shortest factorization we desire; instead, it is the factorization that uses the fewest gate injections. Since the elements of $G_\tau$ are effectively free, we are able to apply as many as we want between gate injections without increasing the total cost of the physical circuit. Using this fact may somewhat simplify the factorization problem, but it is still likely to be computationally difficult and require heuristic methods.

Apart from this brief discussion, we defer addressing the problem of compiling with these codes for future work. Solving this problem in general would be useful for many other codes. For example, the hyperbolic surface code known as Bring's code~\cite{Breuckmann_2024} similarly has a logical Clifford group generated by a non-standard gateset.

\subsection{State preparation}

\subsubsection{$Z$ and $X$ eigenstates}
\label{sec:zx_prep}

Taking inspiration from adaptive syndrome extraction~\cite{berthusen2025} and the state preparation of the many hypercubes code~\cite{Goto_2024}, we present the following fault-tolerant, bare ancilla state procedure for preparing the logical $\ket{0}$ and $\ket{+}$ states for a $C_4 \otimes_\tau \mathfrak{D}(C)$ code. The observation driving this procedure is this: if we can ensure that an odd number of errors are propagated to each $C_4$ block, then measuring the opposite-type $C_4$ generators will tell us whether or not an error has occurred. 

Formalizing this, let $\mathcal{S}_{C_4}$ be the weight-4 generators of $C_4 \otimes \mathfrak{D}(C)$ coming from the $C_4$ blocks, and let $\mathcal{S}_{\mathfrak{D}(C)}$ be the concatenated generators coming from $\mathfrak{D}(C)$. Also let us describe the $i$th such stabilizer generator as $\mathcal{S}^{(i)}_{C_4}$ (resp. $\mathcal{S}^{(i)}_{\mathfrak{D}(C)}$).
After measuring the $i$th concatenated generator $\mathcal{S}^{(i)}_{\mathfrak{D}(C)}$ non fault-tolerantly using the circuit in Fig.~\ref{fig:state_prep}(a) (or Fig.~4 of Ref.~\cite{berthusen2025}), there may be residual errors on the system. In particular, a hook error~\cite{Dennis_2002} on the ancilla qubit may propagate to many errors on the data qubits. However, due to the ordering of the CNOT gates in the syndrome extraction circuit of $\mathcal{S}^{(i)}_{\mathfrak{D}(C)}$, any hook error is guaranteed to propagate as single-qubit errors on several $C_4$ blocks. Furthermore, unless the hook error happened before the first CNOT gate (in which case the error is a stabilizer), then at least one $C_4$ block will only have one error propagated to it. Hence measuring the syndromes of opposite-type $C_4$ generators, e.g with the circuits shown in Fig.~\ref{fig:state_prep}(c)-(d)~\cite{paetznick2024}, will detect that an error has occurred. To slightly reduce the gate count, not all $C_4$ generators need to be measured; instead, only those that share support with $\mathcal{S}^{(i)}_{\mathfrak{D}(C)}$ might have a propagated hook error. The set of $C_4$ generators to measure is then:
\begin{equation}
    \{S^{(j)}_{C_4} ~|~ \textrm{supp}(\mathcal{S}^{(i)}_{\mathfrak{D}(C)}) ~\cap~ \textrm{supp}(\mathcal{S}_{C_4}^{(j)}) \neq \emptyset \},
\end{equation}
where $\textrm{supp}$ denotes the set of qubits on which a stabilizer generator acts non-trivially. If any measured $C_4$ generator yields a -1 measurement result, then we abort the state preparation and try again. In the event that all measured $C_4$ generators yield a +1 measurement result, we can be reasonably certain that an error has not uncontrollably propagated through the system, and we continue measuring the remaining concatenated generators in the same manner until completed. 

\begin{figure}
    \centering
    \includegraphics[width=0.95\linewidth]{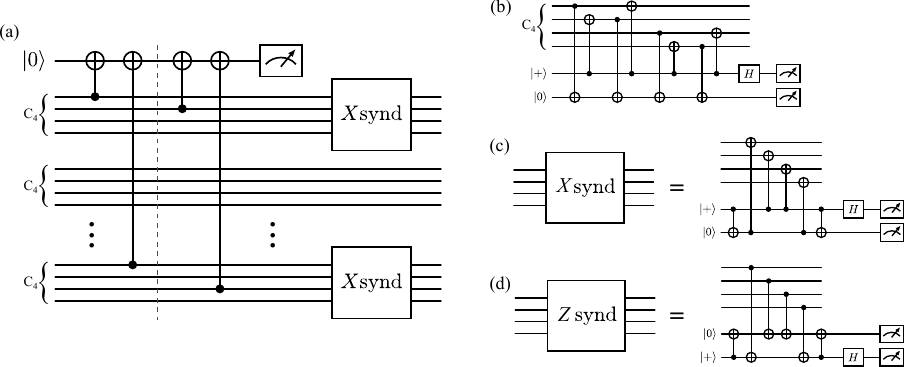}
    \caption{\small (a) State preparation procedure for $C_4$-CSD codes. Concatenated generators are measured in such a way that hook errors propagate to single qubit errors on several $C_4$ code blocks. Measuring the opposite-type $C_4$ generator informs whether a hook error has occurred. (b) Circuit to fault-tolerantly measure both the $XXXX$ and $ZZZZ$ stabilizers simultaneously by using `flagcillas' from Ref.~\cite{Reichardt_2020}. Hook errors propagate to the data qubits but are then caught by the opposite type stabilizer measurement. (c) Circuit to fault-tolerantly measure the $XXXX$ stabilizer of a $C_4$ block using a single additional flag qubit prepared in the $\ket{0}$ state. (d) Circuit to fault-tolerantly measure the $ZZZZ$ stabilizer of a $C_4$ block using a single additional flag qubit prepared in the $\ket{+}$ state. These two circuits were adapted from Ref.~\cite{paetznick2024}.}
    \label{fig:state_prep}
\end{figure}

Despite using two more CNOT gates than the measurement circuits of Fig.~\ref{fig:state_prep}(c)-(d), it might be advantageous to use the `flagcilla'~\cite{Reichardt_2020} $C_4$ syndrome extraction circuit as depicted in Fig.~\ref{fig:state_prep}(b). In this circuit, the measurement ancillas act as flag qubits for hook errors of the opposite type. The benefit of using this syndrome extraction circuit is that both the $XXXX$ and $ZZZZ$ stabilizer measurement results are obtained and can be used for error detection. One set of measurements should be deterministic (e.g. $ZZZZ$ if preparing the $\ket{\overline{0}^{\otimes k}}$ state), while the other will be random the first time measured but deterministic if remeasured. Hence by observing for changes between measurements, we can detect for opposite-type errors which may have occurred. Aborting the state preparation and restarting when detecting such an error could potentially lead to higher fidelity output states.
An important subtlety is that the parallelized circuit cannot be used the first time the non-deterministic $C_4$ generators are measured, as there is no way to know if their measurement results should be treated as a flag. Instead, the individual circuits, Fig.~\ref{fig:state_prep}(c)-(d), must be used for the first measurement, after which the parallel measurement circuit can be used. 

This procedure is fault-tolerant since $q$ faults can produce an outgoing error of at most weight $q$. To see this, consider the measurement of a single, weight $q_{\max}$ concatenated generator $\mathcal{S}^{(i)}_{\mathfrak{D}(C)}$. The worst-case propagated error occurs on the ancilla qubit $q_{\max}/2$ CNOT gates into the circuit. For this error to be undetectable, the $q_{\max}/2$ $C_4$ blocks each have to suffer a measurement error, or $q_{\max}/2$ additional errors have to occur on the same $C_4$ code blocks. In general, a fault occurring during a concatenated generator measurement $t$ CNOT gates into the circuit spreads to (up to stabilizers) at most a weight $\min(t,q_{\max}-t)$ hook error. By splitting the support of the concatenated generator on each $C_4$ block across the middle of the circuit, the hook error has at most weight $1$ support on any given $C_4$ block and is therefore detectable by that generator. So, for any weight $t$ hook error, a further $t$ measurement or data qubit errors are needed on the $C_4$ generator measurements to make that hook error undetectable. The error detection gadgets in panels (b)-(d) are similarly fault-tolerant in the sense that hook errors are detected, and an additional ancilla or measurement error is required for it to go undetected. Furthermore, they, and subsequent measurements of the concatenated generators do not increase the weight of incoming data errors. Hence the weight of the outgoing error is equal to the number of circuit faults during state preparation. However, note that this scheme is only fault-tolerant up to the weight of a stabilizer, which may be smaller than the distance.

\begin{figure}[t]
    \centering
    \includegraphics[width=0.8\linewidth]{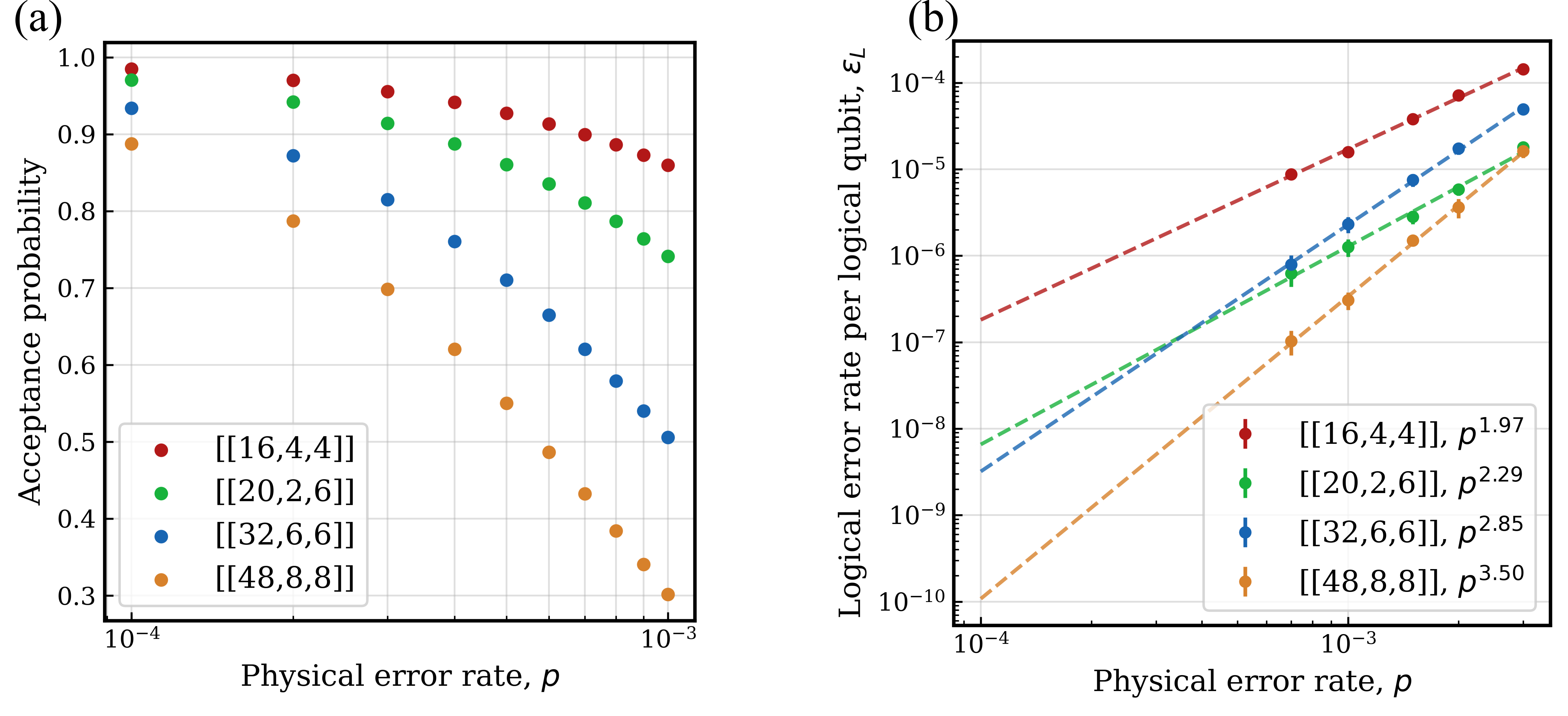}
    \caption{\small (a) Acceptance probability of preparing $\ket{\overline{0}}$ for several $C_4$-CSD codes using the state preparation procedure of Section~\ref{sec:zx_prep}. A prepped state is accepted if the deterministic checks all yield +1 measurement results and the opposite-type checks are consistent each time they are measured. (b) Resulting logical error rate per logical qubit, $\epsilon_L$ of the accepted states.}
    \label{fig:state_prep_sims}
\end{figure}

In Fig.~\ref{fig:state_prep_sims}, we present the results of benchmarking the state preparation procedure. We use a circuit-level noise model as described in Section~\ref{sec:mem_perf} and decode over the entire noisy circuit using the syndrome measurements obtained during the preparation as well as a final round of syndromes obtained from destructively measuring the data qubits.
In panel (a) we plot the rate at which logical states were accepted after attempting state preparation. A state is accepted if all deterministic measurement yielded a +1 result, and all nondeterministic measurements yielded consisted results each time they were measured. At physical error rates near $1\times10^{-4}$, which are being targeted for next-generation quantum hardware, the acceptance rates are high enough that only one or two attempts would typically be required to successfully prepare a codeblock.
In panel (b) we plot the logical error rate per logical qubit, $\epsilon_L = 1 - (1-p_L)^{1/2k}$, where $p_L$ is the probability that there were any logical operators predicted incorrectly. We fit the data to a power law and extract exponents that match reasonably well with the expected value of $\lfloor(d-1)/2\rfloor+1$, with the exception of the $[[20,2,6]]$ code. This discrepancy indicates an error floor with fault distance two for the specific state preparation circuit. To recover the expected performance, it may be sufficient to change the order of the extracted generators or to add more intermediate error detection. Error bars display the standard deviation of $\epsilon_L$, see Section.~\ref{sec:mem_perf}. 
Increasing the blocklength beyond the $[[48,8,8]]$ code did not further drive down the logical error rate. This is likely due to the fact that the state preparation circuits for these large $C_4$-CSD codes are very ($\sim 250-300$ layers) deep, leading to an overwhelming build up of errors due to the data qubits idling. Optimizing the presented state preparation procedure or designing an alternative will be necessary to realize these large codes.

To increase the acceptance rate of the procedure, we could allow for some number of detected errors before exiting and reattempting the preparation. Doing so comes at the cost of a lower fidelity of the resulting state. In Fig.~\ref{fig:state_prep_sims_allow}, we show the acceptance probability and resulting logical fidelity when allowing for some number of errors before postselecting. As expected, allowing for errors effectively reduces the distance of the code, resulting in higher logical error rates; however, the acceptance probability is significantly increased throughout the range of physical error rates.

In achieving this fault-tolerance and performance, we were required to serialize the measurement of the concatenated generators $\mathcal{S}^{(i)}_{\mathfrak{D}(C)}$ in such a way that only a single error could be propagated to each $C_4$ block, resulting in deep state preparation circuits. This is somewhat mitigable as we can measure concatenated generators in parallel which do not have overlapping support:
\begin{equation}
    \textrm{supp}(\mathcal{S}^{(i)}_{\mathfrak{D}(C)}) ~\cap~ \textrm{supp}(\mathcal{S}^{(j)}_{\mathfrak{D}(C)}) = \emptyset
\end{equation}
The concatenation structure and choice of logical operators for the $C_4$ code ensures that if the above condition is satisfied, then $\mathcal{S}^{(i)}_{\mathfrak{D}(C)}$ and $\mathcal{S}^{(j)}_{\mathfrak{D}(C)}$ will have support on different $C_4$ blocks and the condition about propagating at most a single error to each $C_4$ block is maintained. This parallelizability could be optimized by choosing alternative bases for the stabilizer generators. It may also be possible to parallelize measurement of concatenated generators that share support on the same $C_4$ block by introducing additional flag qubits.
Furthermore, it may be possible to reduce the circuit depth by taking advantage of the concatenated structure of $C_4 \otimes_\tau \mathfrak{D}(C)$ and performing some syndrome measurements using $C_4$ ancilla blocks.
Despite this, the presented tailored procedure is competitive compared to other state preparation methods, particularly in terms of ancilla qubits and CNOT gates. Assuming ancilla qubits can be reused, only $n$ additional qubits are needed to fully parallelize the $n$ flagcilla $C_4$ stabilizer measurements in Fig.~\ref{fig:state_prep}(c) (unless more than $n$ concatenated generators can be measured in parallel, which is unlikely). Additionally, the CNOT gate count is nearly optimal compared to a single round of bare ancilla syndrome extraction.

\subsubsection{$\overline{U}_P(\pi)$ resource states}
\label{sec:y_prep}

To supplement $G_\tau$ and implement the full Clifford group, gates from $\langle \overline{U}_P(\pi)\rangle$ or $\langle \overline{U}_{P}(\pi)_{| s}\rangle$ are needed. As discussed prior, it is convenient to implement this gate by first preparing the $\overline{U}_P(\pi) \ket{\overline{+}}^{\otimes 2k}$ state and the performing state injection to apply the gate on a computational code block. Indeed, from Eq.~\eqref{eq:s_tele_circuit} it can be seen that essentially the \textit{only} operation we will be doing is consuming these $\overline{U}_P(\pi) \ket{\overline{+}}^{\otimes 2k}$ states during state injection. Fortunately, preparing them is straightforward given the $Z$ and $X$ eigenstate preparation method described in the previous section: we first prepare the $\ket{\overline{+}}^{\otimes 2k}$ logical state and the apply the circuit to implement $\overline{U}_P(\pi)$ on this ancilla codeblock. Due to the intra-block physical CZ gates, the $\overline{U}_P(\pi)$ gates are not immediately fault-tolerant, as errors can spread throughout the codeblock. To remedy this, we replace each physical CZ gate with the flag gadget~\cite{Chao_2018} below:
\begin{equation}
\begin{quantikz}[background color=white, row sep={0.7cm,between origins}]
\lstick{$\ket{0}$} & \targ{} &  &  & \targ{} &  & \meter{} \\
 & \ctrl{-1} & \ctrl{0} & \ctrl{2} & \ctrl{-1} &  &  \\
 & \targ{} & \ctrl{-1} &  & \targ{} &  &  \\
\lstick{$\ket{+}$} & \ctrl{-1} &  & \ctrl{0} & \ctrl{-1} & \gate{H} & \meter{}      
\end{quantikz}
\label{eq:cz_flag}
\end{equation}
This gadget catches correlated two-qubit errors that may happen during the execution of the intrablock CZ gate. If either of the two flags triggers, we discard the state and attempt another $\overline{U}_P(\pi)$ preparation. It does not catch input errors, but we have guaranteed the input $\ket{\overline{+}}^{\otimes 2k}$ to be sufficiently noiseless. Hence after postselecting on the flag qubits, we will obtain an error scaling of $O(p^{\lfloor (d-1)/2\rfloor + 1})$ for the entire $\overline{U}_P(\pi)$ state preparation procedure as well.

Preparing partial plus states is more difficult, as we cannot use the method of Section~\ref{sec:zx_prep}. The use of partial plus states to perform targeted logical operations was first introduced in Ref.~\cite{Goto_2024}, in which several methods for preparing these states on the many-hypercubes code were introduced. The concatenated nature of $C_4$-CSD codes may allow for straightforward adaptations of the approaches presented in that work. Alternative procedures, such as flag-at-origin~\cite{amaro2025}, can be used to prepare arbitrary CSS states including partial plus states. We leave it to future work to investigate how this state preparation method performs in practice.

\subsection{Quantum error correction}
\label{sec:qec}

The lowest overhead syndrome extraction method is by using bare ancilla qubits. Given their sparse structure, quantum low-density parity-check codes (qLDPC) codes have bare ancilla syndrome extraction circuits that can be done in a constant depth.
Furthermore, many qLDPC codes are single shot~\cite{Bombin2015}, meaning that a single round of syndrome extraction suffices to tolerate syndrome errors. 
Given that CSD codes can, in general, have relatively high weight stabilizers, bare ancilla syndrome extraction may not be the best choice. When the base non-CSS code is LDPC, so too will the concatenated symplectic double code be LDPC. In such cases, bare ancilla syndrome extraction may be a viable option.
It is possible to re-purpose the logical $\ket{0}$ and $\ket{+}$ eigenstate preparation circuits of Section~\ref{sec:zx_prep} to perform bare ancilla syndrome extraction. The only difference is that now both types of stabilizer generators must be measured, leading to syndrome extraction circuits twice as long as the---already deep---state preparation circuits. Furthermore, since these codes are not known to be single shot, the syndrome measurements may need to be repeated, potentially leading to another $O(d)$ increase in the circuit depth. Having the data qubits idle for such a long time during each QEC round will have significant negative effects on the logical error rate as well as the wall-clock time of the logical quantum algorithm. To partially address this and reduce time overheads, we note that adaptive syndrome extraction~\cite{berthusen2025} could be applied. Briefly, the procedure would proceed as follows: first the $XXXX$ and $ZZZZ$ $C_4$ stabilizer generators of every $C_4$ codeblock would be measured using the flagcilla circuit, Fig.~\ref{fig:state_prep}(c). Then depending on the whether or not errors were detected, additional concatenated generators would be measured, following the state preparation procedure or alternative bare-ancilla-style syndrome extraction. 
Applying this adaptive procedure could reduce the number of concatenated generators that are measured which, given the serialized nature in which the concatenated generators must be measured, could serve to significantly reduce the circuit depth.
We leave it to future work to determine the feasibility of performing bare-ancilla syndrome extraction in this way on LDPC instances of CSD codes. 

We avoid the difficulties of bare-ancilla syndrome extraction by instead using Knill-~\cite{Knill_1998} or Steane-style~\cite{Steane1997} syndrome extraction. These methods make use of the fact that transversal CNOT gates implements $\bigotimes_{i=1}^n \overline{CNOT}_{i, i+n}$ on CSS codes. As this operation is transversal, the entire protocol is inherently fault-tolerant. This comes at the cost of having to prepare additional ancilla codeblocks for each syndrome extraction---two for Knill QEC and one for each basis in Steane QEC. While we have fault-tolerant methods to prepare CSD codes in the $\ket{\overline{0}} ^{\otimes 2k}$ or $\ket{\overline{+}}^{\otimes 2k}$ states, the nonzero failure probability and deep circuits imply that many extra codeblocks will be required so that the data qubits are not waiting prohibitively long. For a quantum computer with sufficiently many physical qubits, this tradeoff may be worthwhile. Furthermore, for high-weight codes such as CSD codes and other concatenated codes, these methods are the only practical methods for performing syndrome extraction as they do not directly depend on the stabilizer weight.

Recall that an injected $\overline{U}_P(\pi)$ gate is required to achieve the full Clifford group on a single $C_4$-CSD codeblock. This gate injection can be combined with Knill QEC:


\begin{equation}
\begin{quantikz}[background color=white, row sep={1cm,between origins}]
\lstick{$\ket{\overline\psi}$} & \qwbundle{4n} & &  \ctrl{1} & \gate{H} & \meter{}\wire[d][2]{c} \\
\lstick{$\overline{U}_P(\pi) \ket{\overline +}$} & \qwbundle{4n} & \ctrl{1}  &  \targ{} &&& \meter{}\wire[d][1]{c}      \\
\lstick{$\ket{\overline 0}$} & \qwbundle{4n} & \targ{} & & & \gate{\overline{Z}} & \gate{\overline{X}} & \rstick{$\overline{U}_P(\pi) \ket{\overline\psi}$}      
\end{quantikz} \approx
\begin{quantikz}[background color=white, row sep={1cm,between origins}]
\lstick{$\ket{\overline{\psi}}$} & \gate[3]{QEC + \overline{U}_P(\pi)} & \rstick{$\overline{U}_P(\pi) \ket{\overline\psi}$} \\
 &   \\
 & 
\end{quantikz}
\label{eq:knill_sx}
\end{equation}
Also recall that an arbitrary $2k$ logical Clifford circuit on a single $C_4 \otimes_\tau \mathfrak{D}(C)$ codeblock can be compiled into a circuit alternating between injected logical $\overline{U}_P(\pi)$ gates and layers of physical single-qubit gates, see Eq.~\eqref{eq:s_tele_circuit}. We can instead use the above combined QEC and injection gadget, resulting in error-corrected quantum circuits of the following form. 
\begin{equation}
\begin{quantikz}[background color=white, row sep={1cm,between origins}]
\lstick{$\ket{\overline{\psi}}$} & \gate{U_i^{(1)}} & \gate[3]{QEC + \overline{U}_P(\pi)} & \gate{U_i^{(2)}} & \\
&\wireoverride{n} & & \wireoverride{n}  \\
&\wireoverride{n} & & \wireoverride{n} 
\end{quantikz} ...
\begin{quantikz}[background color=white, row sep={1cm,between origins}]
& \gate[3]{QEC + \overline{U}_P(\pi)} & \gate{U_i^{(m)}} & \rstick{$\overline{U} \ket{\overline\psi}$} \\
& &\wireoverride{n} \\   
& &\wireoverride{n}
\end{quantikz}
\label{eq:s+qec_tele_circuit}
\end{equation}
Again, since the final step of Knill-style QEC is a Pauli correction, the single qubit gates $U_i$ can be folded into this layer, yielding logical Clifford circuits where functionally the only operation happening is Knill-style QEC. 

\subsection{Magic state distillation}
\label{sec:magic}

To complete a universal gateset, a non-Clifford gate is required. These non-Clifford gates are typically not easily accessible on many QECCs, and so methods such as magic state distillation~\cite{bravyi2005}, code-switching~\cite{Anderson_2014}, or magic state cultivation~\cite{gidney2024} are used. For $C_4$-CSD codes, we instead apply the `zero-level distillation' scheme~\cite{Goto2016, hirano2024, itogawa2025}. Briefly, the zero-level distillation scheme functions by first non-fault-tolerantly preparing a logical magic state such as
\begin{equation}
    \ket{H^+} = \cos \frac{\pi}{8}\ket{0} + \sin \frac{\pi}{8} \ket{1} ,\quad \ket{CZ} = \frac{\ket{00}+\ket{10}+\ket{01}}{\sqrt{3}}.
\end{equation}
To confirm that the state was prepared correctly, logical Clifford operators are measured. As suggested by their names, the $\ket{H^+}$ magic state is a +1 eigenstate of the $H$ operator, whereas the $\ket{CZ}$ magic state~\cite{Gupta2024} is a +1 eigenstate of the $CZ$ operator. If the measurement result are different than +1, we abort and retry the procedure until successful. 

In general, performing the measurement of the logical $H$ and $CZ$ operators can be difficult; however, for codes where these operators are transversal, then they can be fault-tolerantly measured using a sequence of physical controlled-Clifford operations. 
As noted in Section~\ref{sec:log_clifford}, every $C_4$-CSD code has a transversal phase gate $\overline{S}_\tau$ that implements some CZ-type logical circuit. For the appropriate choice of logical basis (which we conjecture always exists), the logical action of this gate is to apply $CZ$ gates between pairs of logical qubits. From this structure, it may be possible to zero-level distill one or more logical $\ket{CZ}$ states into a single $C_4$-CSD codeblock. We leave it to future work to determine if this is a reasonable magic state distillation method.


\section{Numerical performance}
\label{sec:performance}

\subsection{Decoders}
\label{sec:decoders}

For qLDPC codes, belief propagation plus ordered statistics decoding (BP+OSD)~\cite{Panteleev_2021, Roffe_2020, Roffe_LDPC_Python_tools_2022} has become the decoder of choice. Concatenated codes often use alternative decoders that take into account the concatenation structure~\cite{c4c6Knill2005, Yoshida_2025, Goto_2024}, and indeed it was shown that performing message passing between concatenation layers yields the optimal decoder~\cite{Poulin_2006}.
Instead, we use the BP+OSD decoder on the detector error model~\cite{Gidney_2021} arising from the noisy state preparation or QEC circuit. We find that this general-purpose decoder provides competitive logical error rates.
We expect that a tailored hard- or soft-decision decoder for $C_4$-CSD codes would perform comparatively to BP+OSD with potentially reduced runtime. Or, given the extremely simple circuit which is decoded over every QEC cycle, it is not inconceivable that a look-up table or similar~\cite{aasen2025} would perform well and be reasonably efficient.

One way to somewhat adapt BP+OSD for decoding concatenated codes is by updating the qubit probabilities before decoding. The concatenated structure of $C_4 \otimes_\tau \mathfrak{D}(C)$ allows us to approximately locate physical errors: for example, observing a -1 measurement for a $C_4$ stabilizer generator tells us that an error is likely to be located on one of those four physical qubits. We can then update the qubit prior probabilities so that the decoder favors assigning errors to the suspected qubits. The syndrome is then decoded in full, as above. 
Another potential improvement is to incorporate error detection and post-selection. In general, for an $[[n,k,d]]$ $C_4$-CSD code, we are expected to be able to correct all errors of up to weight $\lfloor(d-1)/2 \rfloor$. The decoder may be able to correct additional errors of higher weight, but it is not guaranteed to succeed. The concatenated structure of $C_4 \otimes_\tau \mathfrak{D}(C)$ also allows us to easily detect certain high-weight errors: for example, if the syndrome indicates that $\lfloor(d-1)/2\rfloor + 1$ $C_4$ blocks are measured to have a -1 syndrome, we can be reasonably sure that the code has an high-weight error which we have a low chance of decoding correctly. As such, instead of attempting to decode, we abandon that shot or post-select on it. This does not work for all high-weight errors---only those with support on many $C_4$ blocks---but could nonetheless be a low cost method to further drive down the logical error rate. 

\subsection{Memory performance}
\label{sec:mem_perf}

Here we simulate the performance of $C_4$-CSD codes when performing $d$ rounds of noisy Steane QEC. The codes are noiselessly prepared in the $\ket{\overline{0}^{\otimes 2k}}$ state, $d$ rounds of noisy Steane syndrome extraction are applied, and the data qubits are destructively measured out, noiselessly. A detector error model of this circuit then is generated and sampled from. Using the syndrome information, the decoder then attempts to predict whether any logical observables have flipped during the circuit. Decoding is considered a success if the predicted observables match the observables which were reconstructed from the destructive measurement, otherwise decoding is considered a failure. The percentage of successes is called the logical error rate and is denoted by $p_L$. We then calculate the logical error rate per round per logical qubit: $\epsilon_L = 1 - (1-p_L)^{1/kd}$, which is plotted in Fig.~\ref{fig:qec}. Error bars are the standard deviation of $\epsilon_L$, see Appendix B of Ref.~\cite{berthusen2025} or Appendix B of Ref.~\cite{xu2023constantoverhead}. 
We employ the following noise model parameterized by a noise strength $p$: two-qubit gates are followed by a two-qubit depolarizing channel with probability $p$, measurement errors are flipped with probability $p$, single-qubit gates are followed by a depolarizing channel with probability $p/10$, and idle qubits are affected by a depolarizing channel with probability $p/10$. To actually perform the circuit-level simulations, we use \texttt{Stim}~\cite{Gidney_2021}. We use BP+OSD to decode over the entire circuit volume, with a maximum of 200 iterations of product-sum belief-propagation and the order-5 combination-sweep post processing step.

\begin{figure}
    \centering
    \includegraphics[width=0.4\linewidth]{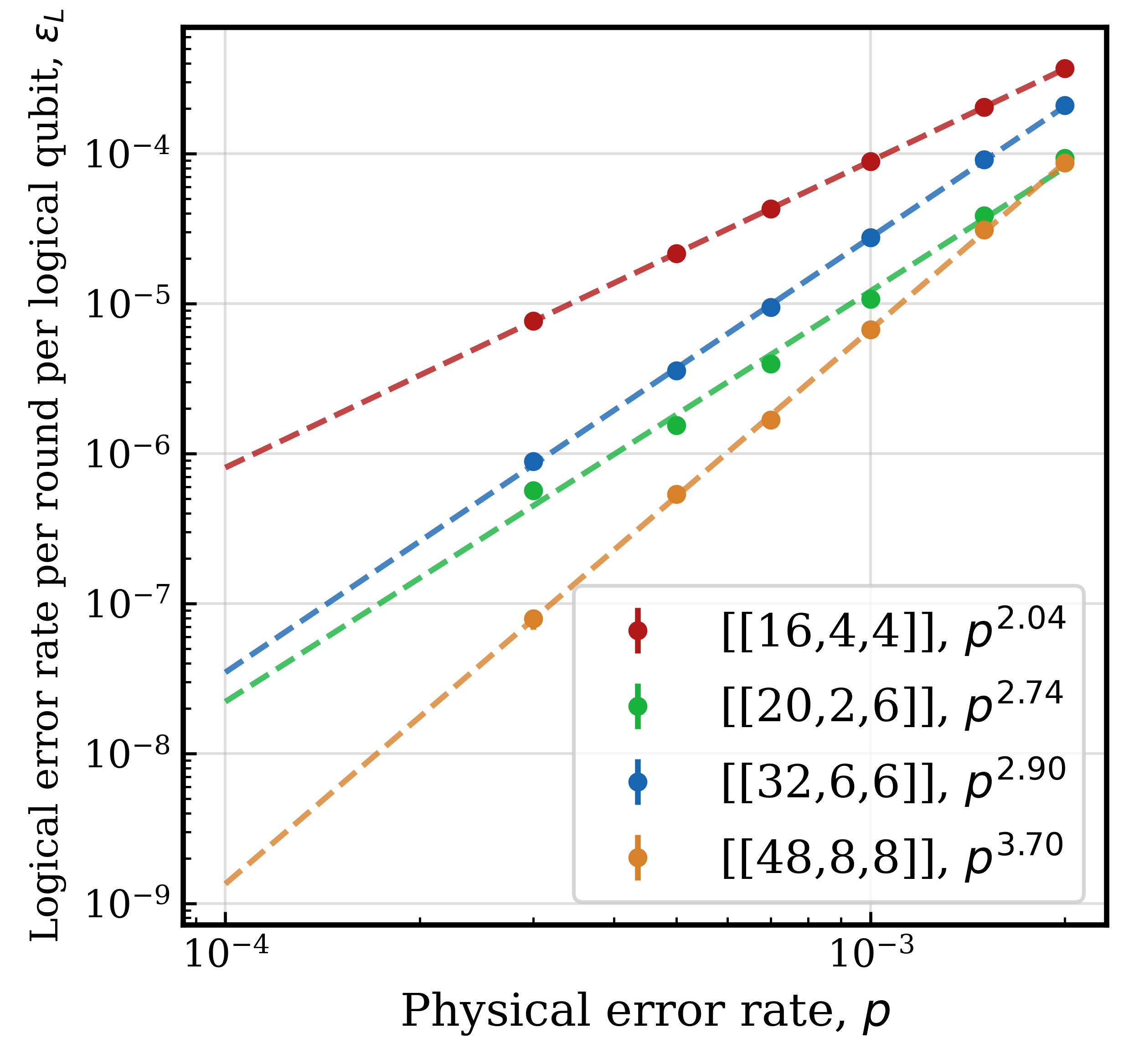}
    \caption{\small Logical error rate per round per logical qubit $\epsilon_L$ as a function of the physical error rate $p$. The circuit-level simulations consist of $d$ rounds of noisy Steane syndrome extraction before noiselessly measuring out the data qubits and constructing a final syndrome. The decoder then uses the entire syndrome volume to predict a logical correction.}
    \label{fig:qec}
\end{figure}

Instead of combining both the state preparation scheme of Section~\ref{sec:zx_prep} and the Steane syndrome extraction in the same simulation, we abstract the simulation as follows: from Fig.~\ref{fig:state_prep_sims}, we have an estimate of the logical error rate of preparing a $C_4$-CSD codeblock. To approximate the error distributions resulting from state preparation, we noiselessly prepare the codestate and then apply a depolarizing channel with strength $p'$ to the data qubits. The noise parameter $p'$ is tuned for each code instance such that the logical error rates of these codestates approximately match those reported in Fig.~\ref{fig:state_prep_sims}. Additionally, we apply a measurement error to the syndrome of the concatenated stabilizer generators with probability $p$. We can safely assume that the syndromes of the $C_4$ generators are error free, since we measure them multiple times and postselect on any differences. 

The results of these simulations are shown in Fig.~\ref{fig:qec}. As expected, the logical error rates are worse than the state preparation fidelities shown in Fig.~\ref{fig:state_prep_sims} by about an order of magnitude---but are nonetheless competitive. We again fit the data to a power law, and now observe that the exponent of most code instances closely follow the expected value of $\lfloor(d-1)/2\rfloor+1$. The $[[20,2,6]]$ code still under-performs due to the aforementioned error floor present in the presented state preparation method. We note that in these simulations we assume to always have an ancilla codeblock available to do Steane QEC, hence avoiding any idle error arising from having to wait for a codeblock to be prepared. On a large quantum computer, this assumption may be valid; however, on a small to medium size computer where physical qubits are limited and ancilla codeblocks are not always immediately available, there may be additional idling times which would result in inferior logical error rates.

\subsection{Logical performance}
\label{sec:log_perf}

Given that the logical circuits are functionally very similar to memory circuits, we expect that replacing Steane QEC with 1-bit Steane QEC combined with gate teleportation will perform comparatively. Typically, it is thought that a logical error rate of $p_L$ allows for the implementation of roughly $1/p_L$ logical gates before the noise in the system becomes overwhelming. This is potentially an overly pessimistic estimation for $C_4$-CSD codes: since an arbitrary-depth logical Clifford circuit can be compiled into a constant number of logical generators (for a fixed $k$), the number of required syndrome extraction rounds is also constant. This implies that we should instead be able to implement roughly $1/p_L$ logical Clifford circuits of any length before decohering. Of course, this analysis becomes more complicated when interleaving non-Clifford gates, but it suggests that $C_4$-CSD codes may provide logical computation benefits over other quantum error correcting codes.

\section{Discussion and outlook}
\label{sec:discussion_outlook}

In this work, we have introduced concatenated symplectic double codes. Given highly symmetric seed non-CSS codes, we are able to construct CSS codes with many fold-transversal gates by using the symplectic double cover construction. Concatenating these symplectic double codes with the $C_4$ code along a $ZX$-duality yields the concatenated symplectic double codes, CSS code for which the fold-transversal gates have been upgraded to SWAP-transversal while keeping the same logical action. The CSD construction is general in that any $[[n,2,d]]$ code $Q$ can be used to construct a $Q \otimes_\tau \mathfrak{D}(C)$ code. We focus on the $C_4$ code for its many SWAP-transversal logical gates; however, there are codes that function similarly. For example, the $[[20,2,6]]$ $C_4$-CSD code discussed in Section~\ref{sec:compilation} is a self-dual code with the same SWAP-transversal gateset as the $C_4$ code. Hence $[[20,2,6]]$-CSD codes are functionally equivalent to $C_4$-CSD codes---with the added complications of more high-weight stabilizers.

We present a method of finding fold-transversal gates of a $C_4$-CSD code (or any self-dual code) by searching for automorphisms of its underlying classical code. Interleaving SWAP- and fold-transversal gates allows us to implement the full Clifford group on all $2k$ logical qubits of a $C_4$-CSD codeblock. When these fold-transversal gates are implemented as state injections, all logical Clifford circuits have a functionally simple physical circuit that alternates between physical single-qubit Clifford gates and fold-transversal gate injections. Furthermore, by combining the gate injections with Knill or Steane QEC, implementing logical Clifford circuits is physically identical to performing syndrome extraction on a quantum memory. We hypothesized that these codes have convenient non-Clifford gates facilitated through zero-level distillation of the logical $\ket{CZ}$ state; however, further work is required to determine whether this method is viable in general. 
Numerically, we find promising circuit-level performance of both the state preparation procedure of Section~\ref{sec:zx_prep} as well as Steane QEC memory experiments.

To be better suited for large-scale quantum computers targeting thousands of logical qubits at logical error rates below $10^{-12}$, optimization to the methods presented in this paper will be required. To achieve those logical error rates, distances greater than the $d=8$ studied here will be required, and hence reworking state preparation to handle larger blocklengths will be necessary. However, $C_4$-CSD codes may already be good choices for medium-scale quantum computers targeting hundreds of logical qubits at logical error rates around $10^{-6}$ to $10^{-8}$. In this regime, the smaller instances of $C_4$-CSD codes are sufficiently well performing, and the compilation problem may not yet be intractable. 
As a proof of concept, it would be interesting to implement a quantum algorithm on a few logical qubits using small $C_4$-CSD codes, such as the $[[16,4,4]]$ or $[[20,2,6]]$ codes. 

There are several open questions to be answered and improvements to be made before $C_4$-CSD codes solidify themselves as a leading contender for fault-tolerant quantum computation:
\begin{enumerate}
    \item What is the best choice for non-CSS seed code $C$? Genon codes~\cite{burton2024} are qLDPC, resulting in qLDPC $C_4$-CSD codes such as the $[[20,2,6]]$ code. However, the parameters of the underlying topological codes are bounded~\cite{Bravyi_2010}, leading to potentially suboptimal parameters in the $C_4$-CSD code. The available logical gateset is equally important: given the fact that automorphism gates on $C$ lift to SWAP-transversal gates on $C_4 \otimes \mathfrak{D}(C)$, highly symmetric seed non-CSS codes are likely preferred.
    \item Can we formulate conditions for when the automorphism group of the non-CSS code $C$ yields $C_4$-CSD codes with an expressive SWAP-transversal gateset? Can we also determine if and when $G = \langle G_\tau, \overline{U}_P(\pi) \rangle$ or $G' = \langle G_\tau, \overline{U}_{\pi |s }\rangle$ is sufficiently powerful to generate the full Clifford group on a single $C_4$-CSD codeblock?
    \item How to most efficiently perform logical compilation on these codes? As we briefly discussed in Section~\ref{sec:compilation}, the compilation problem reduces to a word factorization in the symplectic group $\textrm{Sp}_{4k}(\mathbb{F}_2)$. Due to the exponential growth in the size of the symplectic group, brute-force and exhaustive search methods will become intractable at even a modest number of logical qubits. Heuristic search methods will likely be required to perform real-time compiling. Even with fairly efficient methods, this may practically limit the size of $C_4$-CSD codes which can be efficiently compiled and computed with. Pre-compiling individual Clifford gates and using existing compilation strategies may be more straightforward at the cost of increased circuit depth.
    \item Are there more efficient state preparation methods? We presented fault-tolerant and reasonably efficient methods for preparing $X$ and $Z$ eigenstates; however, the circuit depth is prohibitive. Indeed, the larger $C_4$-CSD code instances were outperformed by smaller distance codes, likely due to the significant idling time. We presented some potential improvements, and there may be entirely different state preparation methods which are more time efficient. Some recently introduced options include the flag-at-origin scheme~\cite{amaro2025} or the partially fault-tolerant methods of Ref.~\cite{reichardt2026}. To realize $C_4$-CSD code instances larger than that which was numerically studied in this work, alternative state preparation methods will likely be required.
    \item Is there a better performing, tailored decoder that makes use of the concatenated nature of CSD codes or their simple physical circuit structure?
\end{enumerate}

More generally, this work and several recent works~\cite{xu2024, reichardt2024, Hong_2024, malcolm2025, berthusen2025_2, yoder2025} have illustrated the utility of automorphism gates in performing logical quantum computation. However, automorphism gates are also limited in their power~\cite{chakraborty2026}, and so developing alternative methods to efficiently find fold-transversal gates will also be useful.
Besides directly looking for codes with large and expressive automorphism groups, it may be beneficial to follow the constructive approach described in this work in which fold-transversal gates are upgraded to be SWAP-transversal through concatenation. Codes such as hyperbolic surface codes~\cite{Breuckmann_2024}, for which the entire logical Clifford group (on a subset of the logical qubits) is generated by fold-transversal gates, may be good candidates for this method.


\section*{Author contributions}

Noah Berthusen led the project and developed the code construction, logical gates theory, and performed the numerical simulations. Elijah Durso-Sabina developed gadgets for the codes. Both authors contributed to the writing and revision of the manuscript.

\section*{Acknowledgments}
The authors thank Simon Burton for answering questions about Ref.~\cite{burton2024}. The authors also acknowledge helpful discussions with Shival Dasu, Andrew Potter, and Serban Cercelescu. We note that during the preparation of this work, Ref.~\cite{kanomata2025} independently discovered the $[[30,6,5]]$ symplectic double code.

\section*{Data availability}

The source code and data to generate the codes, simulations, and figures in the paper are available at Ref.~\cite{csd_code}.

\bibliography{bibliography}

\appendix

\section{Additional figures}

\begin{figure}[H]
    \centering
    \includegraphics[width=0.75\linewidth]{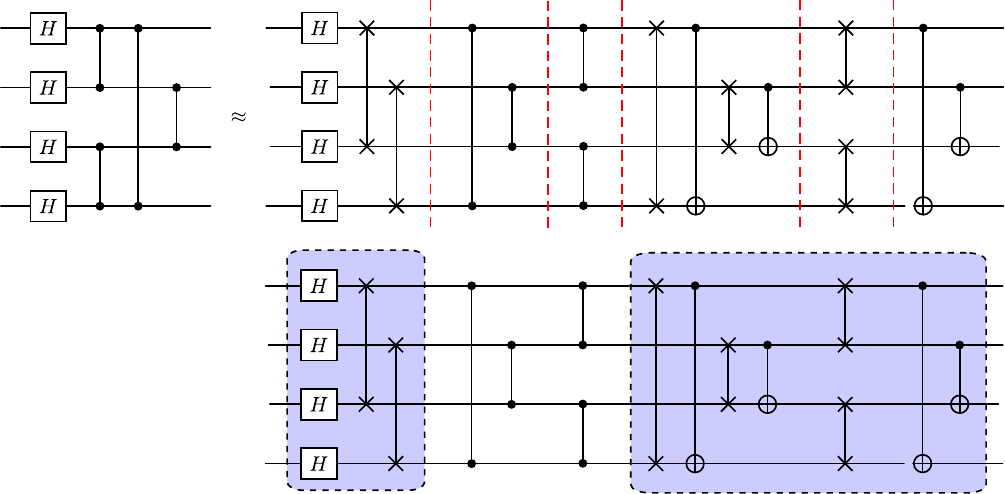}
    \caption{\small Compiling the creation of a graph state (up to Pauli correction and global phase) on the $[[16,4,4]]$ $C_4$-CSD code. Here, we require two fold-transversal gates, not highlighted by the blue boxes. Note that this decomposition as provided by GAP uses the fewest generators, but it may be the case that there are longer sequences that use fewer fold-transversal gates.}
    \label{fig:1644compilation}
\end{figure}

\begin{figure}[H]
    \centering
    \includegraphics[width=0.8\linewidth]{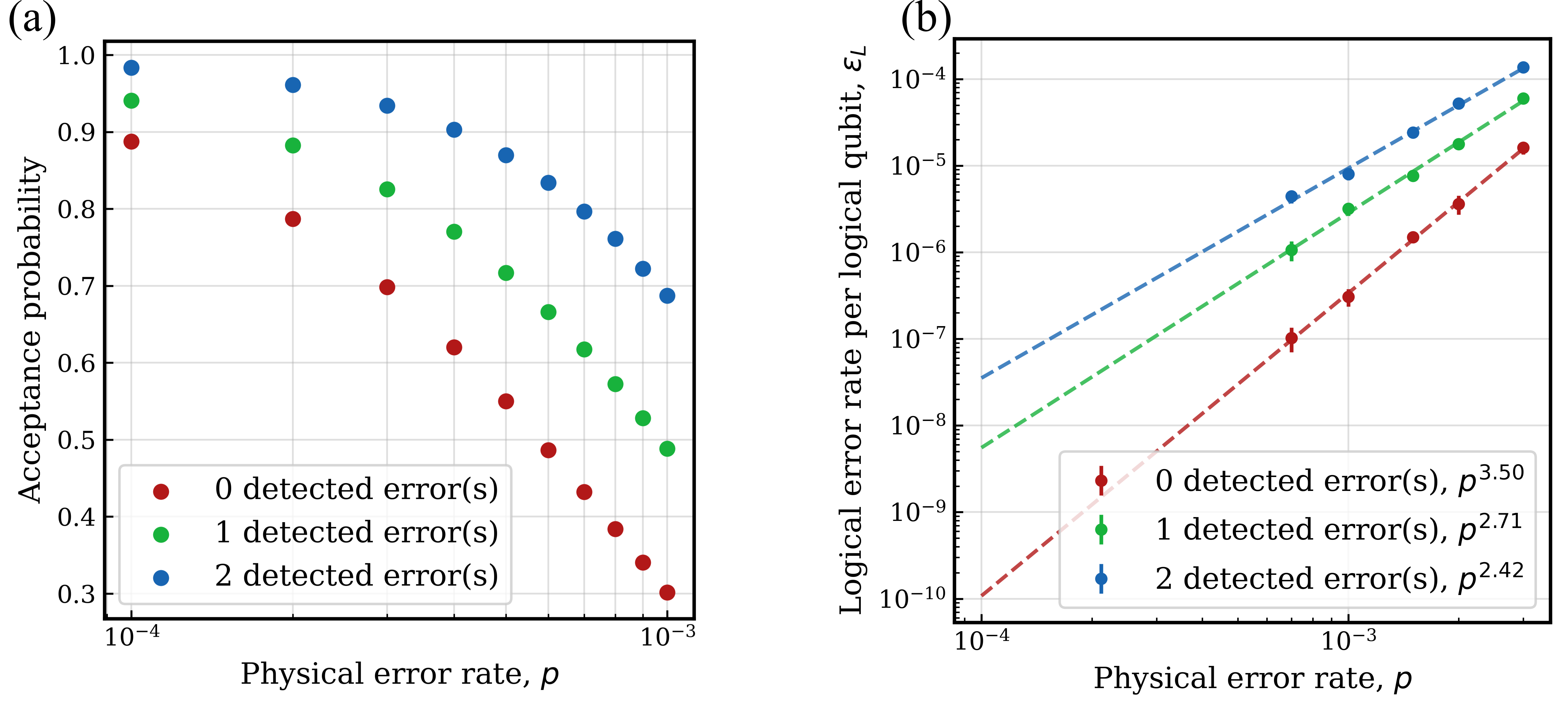}
    \caption{\small Benchmarking the $\ket{\overline{0}}$ state preparation procedure, Section~\ref{sec:zx_prep}, for the $[[48,8,8]]$ $C_4$-CSD code while allowing for some nonzero amount of errors before postselection. (a) Acceptance probability as a function of the physical error rate, $p$. A prepped state is accepted if fewer than $m$ detectors (parities of time adjacent syndrome measurements) trigger. (b) Resulting logical error rate per logical qubit, $\epsilon_L$, of the accepted states.}
    \label{fig:state_prep_sims_allow}
\end{figure}

\newpage
\section{Seed non-CSS codes}
\label{apx:seed_codes}

\begin{align}
H_{[[4,2,2]]} =  \left(\begin{array}{c|c}
1 0 0 1 & 0 1 1 0 \\
0 1 1 0 & 1 0 0 1 
\end{array} \right)
\label{eq:noncss422}
\end{align}

\begin{align}
H_{[[5,1,3]]} =  \left(\begin{array}{c|c}
1 0 0 1 0 & 0 1 1 0 0 \\
0 1 0 0 1 & 0 0 1 1 0 \\
1 0 1 0 0 & 0 0 0 1 1 \\
0 1 0 1 0 & 1 0 0 0 1 
\end{array} \right)
\end{align}



\begin{align}
H_{[[8,3,3]]} =  \left(\begin{array}{c|c}
11111111 & 00000000 \\
00000000 & 11111111 \\  
01011010 & 00001111 \\  
01010101 & 00110011 \\  
01101001 & 01010101
\end{array} \right)
\end{align}

\begin{align}
H_{[[12,4,4]]} =  \left(\begin{array}{c|c}
100101010100 & 010101011101 \\
010101010001 & 001100001010 \\
001100000101 & 101010010111 \\
000011000101 & 010110010111 \\
000000110000 & 110011001100 \\
000000001100 & 111100000011 \\
000000000011 & 001111001111 \\
000000000000 & 000000111111
\end{array} \right)
\end{align}


\end{document}